\documentclass[acmsmall, nonacm]{acmart}

\AtBeginDocument{%
  \providecommand\BibTeX{{%
    \normalfont B\kern-0.5em{\scshape i\kern-0.25em b}\kern-0.8em\TeX}}}

\usepackage[utf8]{inputenc}
\usepackage{url}
\usepackage{amsmath,amsthm}
\usepackage{bbm}

\usepackage{slivkins-setup,nicefrac}

\usepackage[ruled,linesnumbered,noend]{algorithm2e} % For algorithms

\newcommand{\alginit}{}
\usepackage[suppress]{color-edits}% suppress
\addauthor{MB}{red}
\newcommand{\mbc}[1]{{\MBcomment{#1}}}
\newcommand{\mbe}[1]{{\MBedit{#1}}}

\addauthor{BL}{blue}
\addauthor{bl}{blue}

\addauthor{RL}{green}

\addauthor{AS}{magenta}

\newtheorem{theorem}{Theorem}[section]

\newtheorem{lemma}[theorem]{Lemma}
\newtheorem{claim}[theorem]{Claim}

\newtheorem{remark}{Remark}
\newtheorem{definition}{Definition}

\newtheorem{discussion}[theorem]{Discussion}

\renewcommand{\xhdr}[1]{\vspace{1mm} \noindent{\bf #1}}
\renewcommand{\fakeItem}[1][$\bullet$]{\vspace{1mm}{\bf #1}\hspace{2mm}}

\newcommand{\term}[1]{\ensuremath{\mathtt{#1}}\xspace}
\newcommand{\SoW}{\term{SoW}}    % Statement of Work
\newcommand{\SoWs}{\term{SoWs}}
\newcommand{\MainProb}{\term{MainProblem}} % main problem
\newcommand{\RelaxProb}{\term{RelaxedProblem}} % fractional relaxation
\newcommand{\InitInfo}{\term{InitInfo}}
\newcommand{\UpdateInfo}{\term{UpdateInfo}}

\newcommand{\Relax}{\term{Relax}} % \relax(I): relaxed problem instance

\newcommand{\alg}{\term{ALG}}
\newcommand{\opt}{\textsc{OPT}}

\newcommand{\algC}{C'} % reduced C in Alg2

\newcommand{\dem}{c}    % concurrency demands
\newcommand{\birth}{b}  % birth time
\newcommand{\arr}{a}    % arrival time
\newcommand{\dur}{d}    % duration
\newcommand{\finish}{f} % finish (completion) time
\newcommand{\pay}{\pi}  % payment
\newcommand{\util}{u}   % utility

\newcommand{\altpay}{\tilde{\pay}}

\newcommand{\evict}{\mathcal{E}}                % eviction probability
\newcommand{\evictapx}{\overline{\mathcal{E}}}  % estimated evict. prob.
\newcommand{\evictest}{\tilde{\mathcal{E}}}  % estimated evict. prob.

\newcommand{\estU}{\overline{U}} % estimated utility for launch plan computation
\newcommand{\estP}{\evictapx} % estimated success prob

\newcommand{\hist}{\mathcal{H}}

\newcommand{\maxC}{C_{\term{m}}} % concurrency demand
\newcommand{\maxD}{D_{\term{m}}} % duration
\newcommand{\maxV}{H}            % value
\newcommand{\maxS}{S_{\term{m}}} % diff between job's birth and start time

\begin{document}
\title[Truthful Online Scheduling of Cloud Workloads under Uncertainty]
{Truthful Online Scheduling of\\ Cloud Workloads under Uncertainty}
\author{Moshe Babaioff}
\affiliation{%
  \institution{Microsoft Research}
  \city{Herzliya}
  \country{Israel}
}
\email{moshe@microsoft.com}

\author{Ronny Lempel}
\authornote{Most work done while the author was at Microsoft.}
\affiliation{%
  \institution{Google}
  \city{Kirkland}
  \state{WA}
  \country{USA}}
\email{rlempel@google.com}

\author{Brendan Lucier}
\affiliation{%
  \institution{Microsoft Research}
  \city{Cambridge}
  \state{MA}
  \country{USA}}
\email{brlucier@microsoft.com}

\author{Ishai Menache}
\affiliation{%
  \institution{Microsoft Research}
  \city{Redmond}
  \state{WA}
  \country{USA}}
\email{ishai@microsoft.com}

\author{Aleksandrs Slivkins}
\affiliation{%
  \institution{Microsoft Research}
  \city{New York City}
  \state{NY}
  \country{USA}}
\email{slivkins@microsoft.com}

\author{Sam Chiu-wai Wong}
\authornotemark[1]
\affiliation{%
  \institution{Microsoft Research}
  \city{Redmond}
  \state{WA}
  \country{USA}}
\email{sam.cw.wong@gmail.com}

\renewcommand{\shortauthors}{Moshe Babaioff et al.}

\begin{abstract}
Cloud computing customers often submit repeating jobs and computation pipelines on \emph{approximately} regular schedules, with arrival and running times that exhibit variance. % MB: maybe replace ''low variance'' with ''that varies in length" or something like that - but then we might need to also change the later stuff
%This is especially true for training machine learning models, where compute requirements are determined by the availability and quantity of training data at execution time.
%
This pattern, typical of training tasks in machine learning, allows customers to \ASedit{partially} predict future job requirements.
% -- but only partially.
%
%Focusing on this pattern,
%which allows for some advance predictability by the submitting customers,
We develop a model of cloud computing platforms that receive \ASedit{statements of work (SoWs) in an online fashion.
%, in online fashion, statements of work that 
The SoWs} describe future jobs whose arrival times and durations are probabilistic, and whose utility to the submitting agents declines with completion time. The arrival and duration distributions, as well as the utility functions, are considered private customer information and are reported by % the submitting
strategic agents to a scheduler that is optimizing for social welfare.

We design pricing, scheduling, and eviction mechanisms that incentivize truthful reporting of \ASedit{SoWs}. %statements of work.
%
%An important technical challenge is to maintain
%
%% BL: I changed `truthfulness' to `our incentive properties' in the next sentence, since technically our mechanism is not truthful (only approximately truthful).  The sentence above seems okay.
%
%\bledit{Our incentive properties are maintained even when the platform %is overbooked
%becomes saturated due to
%%many jobs simultaneously requiring more resources than predicted.
%%unlikely %low probability
%%events where
%larger-than-anticipated
%realized demand.}
%
%
An important challenge is maintaining incentives despite the possibility of the platform becoming saturated.
We introduce a framework
%Our solution takes the form of a
%general framework that
to reduce
%reduction from
scheduling under uncertainty to a relaxed scheduling problem without uncertainty.
Using this framework,
we tackle both adversarial and stochastic submissions of statements of work,
and obtain logarithmic and constant competitive mechanisms, respectively.
%
%total realized demand being larger than anticipated.}
%joint job realizations.
%
%
%to
%%the problem
%%the mechanism design task
%%to finding
%a posted-price scheduler in a relaxed scheduling environment with no uncertainty.
%in job requirements.
%\mbc{maybe replace by: “We first present a general framework that reduces our problem with uncertainty in job requirements, to the problem of designing a posted-price scheduler for a relaxed scheduling problem with no uncertainty.”}
%
%Using an algorithmic framework based on TODO,
%
%. In the adversarial case, {the competitive ratio of our mechanism
%social welfare achieved by our mechanism
%is logarithmic in the maximum %\blcomment{should we add: in the maximum job value and length \mbc{ratios}}.
%job utility and length.}
%When statements of work arrive stochastically from known distributions, our mechanism is $O(1)$-competitive.
%

\end{abstract}




\vspace{-10mm}\maketitle

\section{Introduction}
\label{sec:intro}
Cloud computing platforms provide computational resources of unparalleled scale to their customers. %Operating from multiple data centers comprised of hundreds of thousands or even millions of processors each, the platforms support many thousands of customers, from small businesses to large enterprises.
Making the most of this increasing scale involves
scheduling the workloads of many customers concurrently using a large supply of cloud resources.
%, but this is a highly complex %non-trivial
%task.
%Customers expect their workloads to be executed quickly and reliably, at a cost-effective price point. Platforms strive for low energy consumption, high utilization, and low hardware failure rates, among others, while delivering service levels that meet customers' expectations. Achieving those service levels becomes difficult when peak demand is high relative to the available resources. At such times, workloads may wait in queues for resources to become available, run on slower hardware, or run in less reliable tiers. Consequently, workloads may complete later than anticipated, resulting in lower utility to customers.
%
Recent years have seen dramatic growth in demand for a particular type of workload: training pipelines for production-grade machine learning models. Such workloads have particular characteristics and challenges that must be addressed by a cloud platform:

\fakeItem\emph{Uncertain Stochastic Job Requirements.} Machine-learned models deployed to production often entail data processing and training pipelines
%(namely, execution DAGs)
that run on a regular schedule, {\eg weekly or hourly.}
%\ASdelete{A model might be retrained weekly, hourly, or at other custom intervals.}
The training step might depend on the completion of several data preparation jobs
(cleaning data, feature engineering and encoding, etc.)
%(joining, filtering and cleaning data, feature engineering and encoding, etc.)
that run in some prerequisite order.  This structure makes it possible to predict future jobs and schedule resources in advance.  However,
%\ASdelete{even if the execution pipeline and schedule are known ahead of time,}
the exact timing of any particular job instance will depend on factors such as the size of the training data
%and the duration of the data-generating process,
%which can vary and
that may only be revealed at the moment the job is to be executed
{(and maybe not even then).}
Thus %\ASdelete{, in the best case,}
the cloud computing system and the customers may learn the \emph{distribution} of a submitted and upcoming job's duration and the time at which it will become available for execution.
    %for each pipeline step. These stacked durations also naturally imply a distribution on the earliest possible launch time of every subsequent step.

\fakeItem\emph{Latency-dependent Utility.} The utility  derived by a customer from each instance of a recurring training pipeline will depend on completion time. In some applications, the earlier the refreshed model becomes available for deployment, the better.
%, as the fresh data it was trained with should cause it to outperform the current production model.
In other cases, customers may have constant utility up to a strict deadline, and lower or no utility past that. The exact sensitivity to completion time varies greatly across customers as it stems from each customer's business problem and model deployment strategy.  %Customers may have constant utility up to a strict deadline, or utilities that smoothly degrade over time, or anything in between.

\fakeItem\emph{Information Asymmetry and Incentives.} Customers are likely to have more information than the cloud platform about their upcoming jobs, as well as the power to manipulate job requirements.  For example, a strategic customer might artificially inflate the size of their training data or introduce unnecessary amounts of concurrency if doing so could result in a
    better price or lower latency.
    We therefore consider these attributes to be private information. The platform,
    %must elicit truthful reporting.
    {which faces strategic agents,
    must incentivize truthful reporting}
    %This is especially important in the context of recurring training jobs, since a regular customer who observes outcomes for different job realizations is more likely to notice if the platform's scheduling and pricing protocols reward strategic behavior.  It is therefore especially important to align the incentives of the customer with the platform, so
    to ensure
    that a customer would not gain an advantage by manipulating the predictions of job requirements.%
    \footnote{{While cloud computing customers tend to submit many workloads, we take the common convention that the customers are myopic, optimizing for each job separately.}}

    %\footnote{In practice, customers of cloud computing submit many workloads to the cloud. We take the common convention that customers are myopic and optimize workload utility on a per-job basis, rather than trying to optimize overall utility across all jobs they own.}

%For a cloud computing scheduler to deliver maximal value to its customers - namely to compute jobs that collectively deliver optimal social welfare - it must be familiar with how the value of jobs depends on their completion times, as well have good predictability regarding future job arrival times and their duration requirements. This paper assumes that these attributes are best known by the customers submitting the jobs to the scheduler, and so considers them as private information. Thus, a truthful mechanism is required to incentivize customers to reveal their jobs' true types to the scheduler.

\vspace{1mm}

How should a cloud platform address these three challenges?
%\mbc{it sounds like "these challenges" refer to incentives, but this is not what comes next. }
Most legacy schedulers are reactive: they have little to no foresight of the arriving workloads,
%and therefore are truly online algorithms,
{deal with jobs as they arrive, and do} not support submission in advance. At the other extreme are schedulers
% where advanced reservation is required: workloads must announce themselves and their full requirements sufficiently in advance, allowing the scheduler to
{that require workloads to announce their requirements sufficiently in advance, so as to better plan their execution.}  Neither approach fully addresses the scenario of machine learning training pipelines where jobs are only partially predictable.

\subsection{Our Contributions and Techniques}

\vspace{-1mm}
\xhdr{A Model for Stochastic Job Requirements.}
%\mbc{maybe add "Our first contribution is the introduction of a model for cloud scheduling in the presence of strategic agents, that capture stochastic job requirements. "}
Our first contribution is a model of cloud scheduling that captures partially predictable job requirements.
%We present a model for truthful online stochastic scheduling of jobs
%We include two types of inherent uncertainty in job requirements: uncertainly in arrival and uncertainty in duration.
In our model, jobs are declared to the scheduling system online. Each job comes with {concurrency demand and} a utility function that determines the value for different completion times, and the scheduler's goal is to maximize total utility (\ie social welfare).
%\ASdelete{Each job's specification includes the value and concurrency information, where}
We assume that the total supply of compute nodes is significantly larger than the concurrency demand of any one job. Importantly, a job's specification also includes a \emph{distribution} over possible arrival times (\ie earliest possible execution time) and {duration (\ie execution}
%job requirements (\ie length of
time needed to complete).  We call this specification a \emph{statement of work} (\SoW).\footnote{\bledit{Our theoretical model suggests an interface where customers declare distributions directly to the platform.  This is an abstraction that highlights customer incentives, since any aspect of the probabilistic information could be manipulated. More generally, a \emph{prediction engine} implemented by the platform could supply some or all of the distributional information on the customer's behalf.  We discuss this further in Section~\ref{sec:conc}.}}
%\rlc{The SoW contains more than just this joint distribution - also consurrency and value - and this informal definition overlooks that.}

The platform can make scheduling decisions and declare prices given advance knowledge afforded by the \SoWs.  However, the
%\ASdelete{true realized}
arrival time of a job is only revealed online at the moment the job arrives, and the true duration of a job may be only partially known until the moment the job is completed.
%Likewise, the true duration of a job \bledit{may only be} revealed at the moment the job is completed.
%\footnote{\bledit{Our model and results also extend to settings where a job's duration is revealed (fully or partially) at the moment it arrives}.}
Jobs are non-preemptable but can be evicted. Payments can depend on realized usage.

The ability to specify a distribution over job requirements instead of reserving resources in advance can significantly impact customer utility.  To give a toy example, suppose that a job submitted at time $0$ will arrive at some (integral) time $k \leq X$, where the probability of arriving at time $k$ is % \mbe{proportional to}
$2^{k-X}$ for each $k < X$ (and otherwise it arrives at time $X$).  The job always requires a single unit of computation for a single unit of time, but needn't be scheduled immediately upon arrival: it provides utility $2^{X-k}$ if it completes in round $k$. If this job is scheduled as soon as it arrives, then it uses only a single unit of computation and the customer's expected utility is $\Theta(X)$.  But if the customer were required to reserve blocks of computation time at the moment of submission, it would be necessary to reserve at least $k$ units to obtain expected utility $\Theta(k)$.
With many such jobs,
%the social welfare lost from
forcing agents to submit deterministic requests {would substantially reduce welfare}
(even though every job is very short).
So it would be significantly advantageous to allow agents to submit probabilistic requests and let the platform allocate only the needed resources at the time they are needed (and charge only for the resources used).
\xhdr{Posted-Price Mechanisms with Eviction.}
We develop a framework for designing truthful online scheduling algorithms for
%jobs with
\SoWs. %that \bledit{assign realization-dependent allocations to jobs at submission time.}
%Our next contribution is a framework for designing truthful online scheduling algorithms for \SoWs.
%in this %stochastic
%setting of adversarial arrivals of statements of work (\SoW), where each \SoW reports stochastic information regarding the job's future arrival and duration.
%
%stochastic job arrival and duration declarations
%\mbc{I do not think ''stochastic setting''  is a good term  to use - we have adversarial arrival of SoWs, each gives stochastic information about jobs. As worded, it sounds like we have a fully Bayesian setting (with SoW arrivals sampled from a prior}\rlc{Moshe, does this read better?}.
%
Our algorithms take the form of posted-price mechanisms that expose a menu of prices, one for each potential allocation of resources.
%
%Our method is dual-based, with an underlying pricing engine that % \mbe{at each time step}
%maintains a price per unit of compute resource at each future time step.  These prices can be viewed as dual variables for a fractional relaxation of the online resource-allocation problem.
The scheduling algorithm can increase these prices over time as new \SoWs arrive to the system.
%\mbc{just to make sure, are these never adjusted only because time have passed (without new arrivals)? RL: to the best of my understanding, no - although we could have adjusted prices based on realizations in progress}
When a \SoW is revealed for an upcoming job, the scheduler will immediately assign an \emph{execution plan} that maps possible arrival times to job start times, chosen to maximize expected customer utility at the current prices (and subject to eviction probabilities, as described below). {Such a scheme incentivizes truthful reporting, since the system optimizes on behalf of the strategic agents.}
This approach has the important benefit that the system can commit, at the moment a \SoW is submitted, to a mapping from realizations to outcome and price.

Posted-price mechanisms for online allocation are not new and have been used in various contexts.  A challenge specific to our setting is that, because job requirements are stochastic, a competitive assignment of execution plans will sometimes inadvertently over-allocate the available supply ex post.
%{in some realizations.}
%due to statistical noise.
A common solution is to leave slack when allocating resources to reduce the chance of over-allocation.
%We employ a similar approach, but
Unfortunately, this does not suffice to address our problem: since our scheduler is intended to run for an arbitrarily long time horizon, even a low-probability over-allocation event will eventually occur, in which case the platform must evict running jobs and/or cancel {future commitments}.
{
It is tempting to simply evict all jobs and reset the system in the (very rare) event that an over-allocation occurs.  However, this extreme policy could have a significant impact on incentives.  A customer who suspects the platform is close to saturation might benefit by misrepresenting their job to finish earlier and thereby avoid an impending eviction.}
%
%\mbe{A central problem we need to address is that the eviction policy has implication on the incentives. If not carefully done, it might create an incentive to misreport the \SoW. Thus,}

Instead, our mechanisms {evicts} jobs in a particular order.  Namely, jobs whose \SoWs arrived most recently are evicted first. This could include jobs that have not yet started executing, which would be cancelled.
%\mbc{maybe expand about removal of active jobs that are \emph{not} running at the over-saturated time, in the spirit of Remark \ref{remark:LIFO-all}?}
This LIFO policy has the important property that the probability a job is evicted is determined at the moment its \SoW is submitted, and is independent of future submissions.  This allows {us}
% our framework
to incorporate eviction probabilities into the choice of execution plan, which is crucial for incentives.  Indeed, we prove that any algorithm that falls within this framework will incentivize truthful revelation of each \SoW, even when the system is close to saturation.

\xhdr{A Reduction to Scheduling Without Uncertainty.}
We provide a reduction framework for designing mechanisms of the form described above.
We consider a relaxed scheduling problem where supply constraints need only hold in expectation over the distribution of job requirements.
%In this relaxed problem
{Hence,}
there is no danger of saturation, so the problem of designing a competitive online algorithm is significantly simpler.  Given an online polytime posted-price algorithm $\alg$ for this relaxed problem, we show how to design a polytime mechanism for the original problem that uses $\alg$ as a guide and (approximately) inherits its performance guarantees.
%This reduction is stated formally in Theorem~\ref{thm:reduction}.

\begin{theorem}[Informal]
Suppose $\alg$ is a robust posted-price online algorithm that is $\alpha$-competitive for the relaxed scheduling problem.  Then for any $\epsilon > 0$, assuming {sufficient supply of compute resources,}
%the total supply of compute resources is sufficiently large,
there is a mechanism in our framework that is $\epsilon$-truthful and $\alpha(1+\epsilon)$-competitive for the {original} scheduling problem.
%with stochastic job requirements.
%\mbc{are the two $\epsilon$ the same? I think the second should be $O(\epsilon)$.}\mbc{discuss computation?}\blcomment{If the two are different, we can make them equal by growing the smaller $\epsilon$.  So the current statement is less exact but cleaner.}
\end{theorem}

%We now briefly describe how our mechanism will make use of the provided algorithm $\alg$.
As \SoWs arrive, our mechanism simulates the progression of $\alg$ (with a relaxed supply constraint) and use $\alg$'s prices when choosing a utility-maximizing execution plan.  Our mechanism also tracks the probability of eviction due to saturation and account for eviction when scheduling.  When the probability of saturation is sufficiently low, the utility-maximizing choice of allocation  approximately coincides between the real and simulated problems.  However, when saturation probabilities become too high, {the true allocation may diverge from the simulation.}
%they may influence the choice of execution plan in a way that desynchronizes the true allocation from the simulation.
%(which needn't account for saturation).
To handle this eventuality, we require that algorithm $\alg$ is \emph{robust}
in the sense that {its welfare degrades gracefully if some job allocations are corrupted by an adversary.}
%it is forced to choose a corrupted, adversarially-selected allocation for some jobs,
%this will not significantly impact its welfare guarantee for the remaining jobs.
%its welfare guarantee degrades gracefully.
In the face of desynchronization, our mechanism will play the role of a corrupting adversary and force $\alg$ to allocate in a manner consistent with the chosen execution plans.
%As long as the probability of saturation is typically very low, the impact of this corruption will be minimal and our mechanism will (approximately) inherit the performance guarantees of the relaxed algorithm $\alg$.

%How do we show that saturation is always unlikely?  One
An important technical challenge is that saturation events are correlated across time. {Indeed, if the system is currently saturated, it is likely to stay near-saturated in the near future,}
%If the system over-allocates and must evict jobs in a given round, then the system is likely close to saturation in subsequent rounds as well,
distorting future allocation decisions.
%as jobs are pushed back to future periods.
In principle this could lead to a thrashing state where over-allocation begets more over-allocation and the system never recovers.
%events, causing a bad state from which the platform does not recover and the probability of saturation is perpetually high.
%See Chawla et al.~\cite{chawla2017stability} for further discussion of this issue.
%We prove this does not occur: regardless of the system state at a given time $t$,
{We rule this out, proving that} the total realized usage quickly returns to concentrating around its expectation.
%, as long as the total supply of resources is sufficiently large relative to the demand of any single job.
The proof involves establishing a novel concentration bound for martingales that may be of independent interest.
%We address this problem by establishing that, starting from any history, resource utilization after sufficiently many steps in the future behaves as a Martingale and hence reverts to a ``good'' state with high probability.

%\xhdr{An Online Algorithm in our Framework.}

%Our framework takes the form of a reduction.  Given an online scheduling algorithm of a particular dual-based form that is designed to work in a setting where job requirements are deterministic, our framework extends its applicability to settings with stochastic job requirements while (approximately) retaining its competitive ratio.  \mbc{can we also say anything about the incentives in the reduction?}

\xhdr{Online Mechanisms in our Framework.}
Finally, we provide two example polytime mechanisms that illustrate how to instantiate our framework.
%In the first example, we consider a worst-case setting
{First, we consider an \emph{adversarial variant},}
where \SoWs arrive online in an adversarial but non-adaptive manner.%
%
%We provide an example \mbc{not just one example, but two, right? We need to talk about the second as well} of how to instantiate this framework.  For this example we consider a worst-case setting where job \SoWs arrive online in an adversarial manner \rlc{do we need to say "adversarial but oblivious/non-adaptive"?} \mbc{I think so}.
\footnote{{Each \SoW still describes a distribution over job requirements and performance is evaluated in expectation over them;
%with respect to those distributions;
it is the \SoW specifications that arrive adversarially.
}}
In this setting, we design an $O(\log (\maxD \maxV))$-competitive
%\mbc{I think writing this as $O(\log D_m +\log H)$ is less ambiguous. Or at least write $O(\log (D_m H))$}
$\epsilon$-truthful online scheduler, where $\maxD$ is the maximum duration of any job
%\mbc{Isn't is the maximum ratio of durations of jobs? As $H$ is a ratio then it seems like talking about ratios everywhere is more consistent (we can later say that (almost) WLOG we call the minimum duration "1" - it is not completely WLOG as we need to know it in advance, right?). Another complication is that starting times should be in multiples of that minimal length - no msmatch in phase. LETS DISCUSS}
{and the positive job values are normalized to lie in $[1,\maxV]$.}
%$\maxV$ is the maximum ratio between the largest and smallest non-zero utility attainable from any job.
This mirrors known logarithmic-competitive online algorithms for resource allocation,
%for the deterministic online scheduling problem
based on exponentially-increasing price thresholds. Indeed we use such an algorithm {as an ``input" to our reduction.}
%as the basis for applying our reduction framework.

%\bledit{
%\begin{theorem}[Informal]
%In the fully adversarial setting, there is an $\epsilon$-truthful $O(\log(\maxD H))$-competitive online mechanism for scheduling with stochastic job requirements.\mbc{discuss computation?}
%\end{theorem}
%}

{Second,} we consider a \emph{stochastic variant} where the jobs' arrival times are arbitrary but their \SoWs are drawn independently from known (but not necessarily identical) distributions.
%\ASmargincomment{clarify this in prelims \& Sec5}
{(Since each \SoW includes a distribution, the prior information is a distribution over distributions.)}
%In this setting,
{We apply} our reduction to a variation on a recent $O(1)$-competitive posted-price mechanism for interval scheduling \cite{chawla2019}, obtaining
an $O(1)$-competitive $\epsilon$-truthful online scheduler.  {Unlike the first example, this mechanism will require that a job's duration is revealed at arrival time (i.e., when ready for execution).}

%\bledit{
%\begin{theorem}[Informal]
%\label{thm:bayesian.informal}
%In the stochastic setting, there is an $\epsilon$-truthful $O(1)$-competitive online mechanism for scheduling with stochastic job requirements. \mbc{discuss computation?}
%\end{theorem}
%}

\begin{theorem}[Informal]
\label{thm:combined.informal}
\mbe{For any $\epsilon>0$,} assuming a {sufficient supply of compute resources,} there is an $\epsilon$-truthful, $\alpha$-competitive online mechanism for scheduling with stochastic job requirements, with $\alpha=O(\log(\maxD H))$ for the adversarial variant, and $\alpha = O(1)$ for the stochastic variant.  {The mechanism for the stochastic variant assumes that job durations are revealed upon job arrival.}
\end{theorem}

\subsection{Related Work}
\label{sec:related}
%In this section,
%we outline the literature \mbe{most} relevant to this work.
%%we outline the literature relevant to this work.
%We start in Section \ref{subsec:systems_lit} by  surveying cloud-inspired systems that are related to our model. We then discuss in Section \ref{subsec:theory_lit} some relevant theoretical results.

%\subsection{Cloud resource management and pricing} \label{subsec:systems_lit}

% SLAs and profiling

\xhdr{Cloud resource management. } Cloud resource management has been a very active research topic over more than a decade.
%Of specific relevance to our work is
{Especially relevant is}
the idea of providing jobs with some form of performance guarantees, often termed Service Level Agreements (SLAs) \cite{curino2014reservation,jyothi2016morpheus,tumanov2016tetrisched}.
  %The SLAs typically include a provider commitment on the compute resources that would be allocated to the job, as well as some form of guarantee for a timely execution (\eg by a certain deadline).
  %
  To enable the SLAs, the system profiles jobs to estimate their resource requirements, duration, and sometimes even infer their deadlines based on data; see, \eg \cite{ferguson2012jockey,jyothi2016morpheus,chung2020unearthing} and references therein. %Because of inherent variations in the size of the input data and the actual execution path, job duration is non-deterministic and may be plausibly assumed to have an empirical distribution. Similarly, the arrival of the job (namely, the time in which a job is ready to be executed) depends on several factors, such as the completion of parent jobs \cite{jyothi2016morpheus,chung2020unearthing}; hence, the arrival of the job is also stochastic.

%Recently, there is a large body of work
{Much recent work} is dedicated to scheduling machine learning workloads.
%Of specific relevance to our work
{Particularly relevant}
are scheduling systems that rely on predicting certain job properties. To highlight a few such systems, Tiresias \cite{gu2019tiresias} is a practical system for scheduling ML jobs on GPUs. It uses estimated probabilities of resource consumption and job completion times to prioritize resource allocation. %
%\ASdelete{through some form of a Gittens-index rule.}
Optimus \cite{peng2018optimus} is a job scheduler for deep learning clusters, which builds performance models for estimating the training time as the function of allocated resources. The models are then used to allocate resources dynamically via a centralized optimization problem, {to minimize}
%with the objective of minimizing
the total completion time. Unlike our model, {these}
%the above mentioned
schedulers do not account for future jobs and do not address incentives. %considerations.

% uses
% online fitting to predict model convergence during training, and
% sets up performance models to accurately estimate training speed
% as a function of allocated resources in each job. Based on the models,
% a simple yet effective method is designed and used for dynamically
% allocating resources and placing deep learning tasks to minimize
% job completion time.

\xhdr{Cloud Pricing.}
The emergence of the cloud business has naturally drawn attention to a variety of economic considerations. Some of the main studied topics include designing proper pricing mechanism (\eg price structure that leads to efficiency or profit maximization, but still is simple and comprehensible to the end user), how to maximize return on investment (\eg through spot pricing~\cite{agmon2013deconstructing, menache2014demand}), how to exploit data for refining the pricing mechanism parameters, etc; see \cite{wu2019cloud,al2013cloud} for surveys on cloud pricing.

{Let us focus on pricing SLAs between job owners}
%We focus our discussion here on a body of works that consider pricing of SLAs between the users (typically, job owners)
and the cloud provider.
%Several papers \cite{jain2014truthful,jain2015near,lucier2013efficient, azar2015truthful} study a model where
{\cite{jain2014truthful,jain2015near,lucier2013efficient, azar2015truthful} posit that} jobs have a certain demand for compute resources and a deadline (or more generally, a value function for completion, as in our model),
{and}
%The goal in these papers was to
design incentive-compatible pricing schemes which maximize the social welfare.  %Another body of work
{\cite{babaioff2017era,jalaparti2016dynamic} take} a  market approach, without zooming in on specific customer.
%Demand estimation is global and based on historical data, and prices are set in advance to guide online allocation.
Our paper differs from all these works by explicitly considering the stochastic setting, where both job arrivals and durations are random, a model that is more relevant for ML jobs.

Pricing for machine learning workloads is a relatively new research area. \cite{bao2018online} proposes a primal-dual framework for scheduling training jobs, where the resource prices are the dual variables of the framework. Other recent works \cite{mahajan2020themis,chaudhary2020balancing} consider auction-based mechanism for scheduling GPUs, with the general goal of balancing efficiency {and} fairness. None of these papers models explicitly the stochasticity in job arrival and duration.

\xhdr{Online Scheduling.}
%There is a rich algorithmic literature
{A rich literature on online scheduling algorithms studies} adversarially-chosen jobs that arrive concurrently with execution, and a scheduling algorithm must choose online which jobs to admit. %A common model associates jobs with a value, length, and deadline before which it must complete.  We focus \mbc{this is unclear. Do you mean that the survey of results below focuses on this? In our model we do not restart evicted jobs. Maybe say something like "We consider a setting with no preemption, but we allow eviction (termination of jobs in execution, losing all used resources). Similar settings where considered in the following literature. "}
%on settings where jobs can be evicted
%%preempted \mbc{I think "preempted" is not the term we should use, as (as far as I know) it refers to systems where there is an option to resume the job later - an assumption we do not make (and do not need for our result - so this wording make them look weaker than they are). I think we better say "we assume and running job can be terminated". If you accept this, note that this is a general comment (so we need to make sure we do not use the term "preempted" elsewhere)} (interrupted)
%but must subsequently be restarted from the beginning if at all.
When job values are related to job length, such as when value densities (value over length) are fixed or have bounded ratio, constant-competitive approximations are possible and can be made truthful~\cite{koren1992d,porter2004mechanism}.
When values are arbitrary, there is a lower bound on the power of any randomized scheduler that is polylogarithmic in either (a) the ratio between longest and shortest job lengths, or (b) the ratio between minimum and maximum job values~\cite{canetti1998bounding}, and such bounds have been matched for certain special classes of job values~\cite{awerbuch1993throughput}.
We likewise obtain a logarithmic approximation, using
%A common approach for obtaining such bounds is to employ thresholds or prices on resources that
resource prices that grow exponentially with usage; similar methods were used in incentive compatible online resource allocation going back to~\cite{bartal2003incentive}, and prior to that as an algorithmic method for online routing~\cite{leighton1995fast,plotkin1995fast} and load balancing~\cite{aspnes1993line,azar1997line}.  %We employ a similar approach when designing a scheduler for our adversarial setting.

A technical challenge in the present paper is to limit the impact of cascading failures that arise because overallocation in one round causes an increase of demand for another round.  A similar challenge is faced by Chawla et al.~\cite{chawla2017stability}, who consider a
%model where jobs have values and deadlines, drawn i.i.d.\ from a fixed distribution.  An
setting where an online scheduler sets a schedule of time-dependent resource prices, and each job is scheduled into the cheapest available timeslot before its deadline.
\bledit{A primary difficulty in this setting is maintaining truthfulness, and further work also explores ways to maintain truthfulness in stateful online resource allocation~\cite{chawla2017truth,devanur2019near,emek2020stateful}.}
%The resulting scheduler has competitive ratio that tends to $1$ as the supply of resources each round grows large.
Another closely related scheduling mechanism appears in Chawla et al.~\cite{chawla2019}, who %model the task of scheduling jobs as allocating time intervals of resources. They
consider a Bayesian setting where job requirements are drawn from known distributions and construct a posted-price $O(1)$-competitive mechanism.  Relative to these papers, our model introduces an extra degree of stochasticity where the submitted job requirements are themselves probabilistic.

\section{Our Model}
\label{sec:prelim}

We consider an idealized model of a cloud computing platform which captures the challenges discussed above. The platform has $C$ homogeneous computation units called \emph{nodes}. Time proceeds in discrete time-steps (or \emph{rounds}), with $t$ denoting a time-step. At each round, each node can be allocated to some job, for the entire round. There is a finite, {known} time horizon $T$.\footnote{The finiteness of $T$ is for convenience when defining problem instances. Our guarantees and analysis will not depend on {the size of $T$}.} The platform interacts with self-interested job owners, called \emph{agents}. Each agent owns exactly one job; we use index $j$ to denote both.

Each job $j$ requires a fixed number $\dem_j$ of nodes during its execution, called \emph{concurrency demand}. A job cannot run with fewer nodes nor benefit from additional nodes. The job \emph{arrives} (becomes ready to execute), at \emph{arrival time} $\arr_j$. Job $j$ that starts running at time $t\geq \arr_j$ will be \emph{in execution} and use $\dem_j$ nodes at each of the $\dur_j$ consecutive time-steps starting at time $t$ {(where $\dur_j$ is called \emph{duration}). If not interrupted,} the job \emph{completes} successfully at time $\finish_j = t+\dur_j$.
% OLD: The job will complete after executing for a duration of $\dur_j$ consecutive time-steps.
Jobs are \emph{non-preemptable}: they must run continuously in order to finish. The scheduler can \emph{evict} a running job at any time, terminating its execution and reclaiming its nodes.\footnote{It is possible to reschedule an evicted job, but our mechanisms and benchmarks will not.  We  therefore treat evictions as permanent for convenience.}
%\ASedit{An evicted job fails completely, and cannot be resumed.}%
%\footnote{\ASedit{We could reschedule an evicted job to start at a later time, but we choose not to.}}

Each job $j$ brings some value to its owner, depending on whether and when it is completed. This value is $V_j(\finish_j)\geq 0$ if the job successfully completes (finishes)
%\ASedit{(\ie just finished executing)}
at time $\finish_j$, for some non-increasing function $V_j(\cdot)$ called the \emph{value function}. Otherwise (\ie if the job is evicted or never starts running) the value is $0$. By convention, the completion time is $\finish_j=\infty$ if the job never completes, and $V(\infty) = 0$.

%\mbc{this makes it look like there is one value, but it is actually a function of completion time.}.
%An evicted job fails entirely, with $v_j=0$. If the job successfully completes, say at time $\finish_j$, its value depends on the completion time in a non-increasing fashion; we denote it $v_j = V_j(\finish_j)$ and refer to $V_j(\cdot)$ as the \emph{value function}.

%An agent $j$ learns about its job in some round $\birth_j$ called the \emph{birth} of $j$.
Each job $j$ is submitted at some time
% BJL: changed ``round'' to ``time'' in case we want to make this continuous
%in some round \mbc{maybe better be consistent and always use "time-step" and not "round"?}
$\birth_j$ called the \emph{birth} or submission time.
At this time, the agent knows the concurrency demand $\dem_j$ and the value function $V_j$, but not the arrival time $\arr_j$ nor the duration $\dur_j$.
%\mbe{At birth time,} the job's arrival time $\arr_j$ and duration $\dur_j$ (the number of execution rounds needed to completion) {are still uncertain,} and cannot be affected by the agent or the platform in any way.
However, the agent knows the joint distribution of $(\arr_j,\dur_j)$, denoted by $P_j$. %No other information about $\arr_j$ (resp., $\dur_j$) is revealed until the job actually arrives (resp., completes), and these cannot be affected by the agent or the platform.
{No other information is revealed until the job actually arrives, at which time the platform learns $\arr_j$.  It may also learn some information about $\dur_j$ when the job arrives, in the form of an observed signal $\sigma_j = \sigma(\arr_j, \dur_j)$.}\footnote{For example, if $\sigma_j = \dur_j$ then the platform learns the duration when the job arrives, and if $\sigma_j$ is a constant then the platform gains no information about duration.  Most of our results hold for arbitrary signals, with the exception of Theorem~\ref{thm:main.stochastic}.}
%that is determined by $\arr_j$ and $\dur_j$.
No further information about $\dur_j$ is revealed until the job completes.

Thus, at the job birth time $\birth_j$ the agent knows the tuple
    $\SoW^*_j = (V_j,\dem_j,P_j)$,
called the true \emph{Statement of Work} (\SoW), and reports a tuple
%At the same round, the agent reports some triple
    $\SoW_j = (V'_j,\, \dem'_j,\, P'_j)$,
with the same semantics as the true \SoW,
{called the \emph{reported \SoW}.} Since our mechanisms incentivize agents to report their true \SoWs, it will be the case that $\SoW_j = \SoW^*_j$.
%\footnote{Our mechanisms do not necessarily incentivize agents to report their \SoWs as soon as possible, so it may be possible for an agent to benefit by delaying a submission.But even with strategic delays, our performance guarantees would continue to hold because our benchmark depends on the \SoWs but not on the birth times.}
%throughout most of the paper.

The number of jobs, their birth times, and their true \SoWs constitute the \emph{birth sequence}.  The birth sequence is initially unknown to the agents and the platform.  We will design schedulers {for two settings: adversarial and stochastic. In the \emph{adversarial variant},} the birth sequence is chosen fully adversarially.  {In the \emph{stochastic variant},} each $\SoW^*_j$ is drawn independently from a publicly known {(and not necessarily identical)} distribution and birth times are chosen adversarially.  {For a unified presentation, we formally define a problem instance as a distribution over birth sequences.}\footnote{{An adversarial choice of the birth sequence corresponds to an unknown point distribution. In the stochastic variant, the distribution is partially known to the platform.}}
%{we posit a distribution from which the birth sequence is drawn, and we define a problem instance as this distribution.}
%\ASedit{It is fixed ahead of time, and is not known initially, neither to the agents nor to the platform. In our main result, the birth sequence is chosen adversarially. An extension assumes that the birth sequence is drawn from some distribution over birth sequences, which is not known to the agents. For a unified presentation, we define a problem instance as a distribution from which the birth sequence is drawn.}

%\ASedit{If job $j$ is successfully completed, at some time $\finish_j$, the agent is charged some amount of money $\pay_j\geq 0$. The agent's utility is then
%    $\util_j = V_j(\finish_j) - \pay_j$.
%If the job never completes, the payment (and the utility) is zero. All agents are risk-neutral (evaluating a distribution of utilities at the expectation)}.
%\blcomment{I think that imposing this as a requirement is too strong.  E.g., it would be reasonable for a scheduler to plan to evict a job after it has run long enough that no value would be obtained, but still charge the job for resources used.}

%\blcomment{While it's true that our scheduler doesn't charge payments upon eviction, this shouldn't be a requirement of the model.  Changed the below to reflect this.}
Once a job $j$ has completed (at time $\finish_j$) or been evicted, the agent is charged a payment of $\pay_j\geq 0$. The agent's utility is
    $\util_j = V_j(\finish_j) - \pay_j$.
A job that is never allocated resources has payment (and utility) zero.
Agents are risk-neutral and wish to maximize their expected utility.
%(evaluating a distribution of utilities at the expectation).
{Our mechanisms are \emph{$\epsilon$-truthful} for some {(small) $\epsilon \geq 0$, meaning that} for each agent $j$, given any reported \SoWs of the other agents and any realization of the other agents' job requirements (durations and arrival times),
%for each problem instance,
agent $j$ maximizes her expected utility by submitting her true \SoW, up to an additive {utility} loss of at most $\epsilon$.
%Here
The expectation is over the realization of $(\arr_j,\dur_j)$.}
%{Given this property, we posit that the agents report truthfully: $\SoW_j = \SoW^*_j$.}

%We ensure that truthful reporting of the \SoW is a dominant strategy for each agent. \AScomment{spell this out more carefully? I'd use notation like $u_j$ for utility and $\pi_j$ for payment. Also, the agent's expectation should be conditioned on all that the agent knows, i.e., the birth time and the SoW.} \blcomment{Note that our mechanism is not truthful.}

%\ASedit{Assuming truthful reporting,}
Our mechanism's performance objective is to maximize the total value (or \emph{welfare})
    $\sum_{\text{jobs $j$}} {V_j(\finish_j)}$.
%under truthful reporting. MB: this is not needed.
%We interpret the total value as \emph{welfare}: the total utility of the agents and the platform.%
%\footnote{Note that the mechanism's utility is the total payment, so payments cancel out.}
We are interested in \emph{expected} welfare,
%for a given problem instance,
where the expectation is taken over all applicable randomness.
%: the realization of arrival times and job durations, the mechanism's random seed, and the birth sequence.

For comparison, we consider {the welfare-maximizing schedule in hindsight,} given the value functions, arrival times, and durations of all jobs. The \emph{offline benchmark} is the expected welfare of this schedule on a given problem instance.
%\ASedit{A very demanding benchmark, it ignores the informational and computational limitations faced by the mechanism.}
We are interested in the competitive ratio against this benchmark. Our mechanism is called \emph{$\alpha$-competitive}, $\alpha \geq 1$, if its expected welfare is at least $1/\alpha$ of the offline benchmark for each problem instance.

%For $\alpha \geq 1$, online algorithm $ALG$ is $\alpha$-competitive if, for any sequence of job \SoWs, the expected welfare obtained by $ALG$ is at least $1/\alpha$ of the expected maximum welfare attainable in hindsight.  Here the expectations are with respect to the arrival times and durations of all jobs.

% \textbf{[BL: The assumption of at-most-one-job-per-round is so we don't have to distinguish between prices offered to different jobs arriving at the same time step.]}\textbf{[RL: Would it work to say that if a batch of jobs arrives at the same time, the system randomizes their order and schedules them in FIFO?.]}

\iffalse
\xhdr{Discussion.}
Our mechanisms do not necessarily incentivize agents to report their \SoWs as soon as possible, so it may be possible for an agent to benefit by delaying a submission. But even with strategic delays, our performance guarantees would continue to hold because our benchmark depends on the \SoWs but not on the birth times. \mbc{only if submission is before earliest arrival. REMOVE}
\fi

\xhdr{Technical assumptions.}
We posit some known upper bounds on the jobs' properties: all concurrency demands $\dem_j$ are at most $\maxC$, all job durations $\dur_j$ are at most $\maxD$, and all values $V_j(\cdot)$ are at most $\maxV$.
%, and each job arrives at most $\maxA$ rounds after its birth.
%\begin{align}\label{eq:assns}
%    \dem_j\leq \maxC,\;
%    \dur_j\leq \maxD,\;
%    \arr_j-\birth_j\leq \maxA
%    \;\text{and}\;
%    V_j(\cdot)\leq \maxV,
%\end{align}
%for some known $\maxC,\maxD,\maxA,\maxV$.
Moreover, $V_j(\birth_j+\maxS)=0$ for some known $\maxS$; in words, each agent's value goes down to zero in at most $\maxS$ rounds after the job's birth. We assume
that $V_j(\cdot)\geq 1$ when positive,%
\footnote{We use this assumption to achieve a multiplicative competitive ratio.
%We could omit this assumption, in which case
Otherwise, our welfare results would be subject to an extra additive loss. This welfare loss would correspond to that of excluding all jobs with sufficiently small values.}
\ie
    $V_j(\cdot)\in \{0\}\cup [1,\maxV]$.
%For notational convenience, at most one job arrives at each time $t$.

%    \bledit{We also assume there is an upper bound $S_{max}$ on the difference between a job's submission \mbc{before we defined it to be called "birth time", we need to pick one and be consistent} time $B_j$ and the latest start time such that the resulting job's value can be non-zero.}

To simplify notation, we assume that at most one job is submitted at each round. Our algorithm and analysis easily extend to multiple submissions per round, modulo the notation; see Remark~\ref{rem:multiple-births}.

The main notations are summarized in Appendix~\ref{app:notation}.

%\begin{tabular}{l|l}
%$C$         & system's computational capacity\\
%$\dem_j$    &concurrency demand of job $j$\\
%$\birth_j$  &birth time of job $j$\\
%$\arr_j$    &arrival time of job $j$\\
%$\dur_j$    &duration of job $j$\\
%$\sigma_j$  &signal about job $j$ duration revealed at arrival time\\
%$P_j$       &joint probability distribution over $(\arr_j, \dur_j)$\\
%$\finish_j$ &completion/finish time of job $j$\\
%$V_j(t)$    & value derived by job $j$ if it completes at time $t$ \\
%$\maxC$     & $\max$ concurrency demand of any job.\\
%$\maxD$     & $\max$ duration of any job \\
%$\maxS$     & $\max$ time difference between birth and completion\\
%$\maxV$     & $\max$ value of any job
%\end{tabular}

%\input{prelim.tex}
%\section{A Framework for Price-Based Scheduling}
\section{The General Framework}
\label{sec:outline}

{This section presents a general framework for our scheduling mechanism (Algorithm~\ref{alg.lifo}), and establishes incentive properties common to all mechanisms in this framework.}

%This section presents the high-level design principles of our scheduling mechanism.  We also establish incentive properties common to all solutions in our framework.  The scheduler is summarized as Algorithm~\ref{alg.lifo}.

%We show how to reduce the design problem propose concrete instantiations of this framework in Sections~\ref{sec:oracle} and~\ref{sec:approximate}.

%The basic protocol and contract between the scheduler and a Statement of Work is as follows.  Upon receiving $\SoW_j$, the scheduler responds with a {\em launch plan} $L_j$ for job $j$. $L_j$ is a mapping between every possible arrival time of job $j$ (there are at most $A_m$ such times) and an execution start time that is the same or later.  This launch plan will be chosen in a particular way that we describe below.  But first we'll describe two pieces of publicly-exposed information that will be used to guide the algorithm (and that customers can observe to verify the procedure).

\begin{algorithm}[t]
\alginit
\caption{Algorithm \textsc{SchedulerFramework}
\label{alg.lifo}}
\setcounter{AlgoLine}{0}
%	\SetKwInOut{Require}{require}
%	\SetKwInOut{Input}{input}
%	\Require{Capacity constraint $C' \geq 1$}
%	\Input{Online arrivals of jobs}
	
%	\BlankLine
%    Initialize $A \leftarrow \emptyset$, $E \leftarrow \emptyset$\;
    %$\pay^1$ and $\evictapx^1$
%    Initialize announced info
%    \ASedit{$(\pay^1, \evictapx^1)$: call $\InitInfo()$}\;         \label{alg:line:init}
    {Initialize: $(\pay^1, \evictapx^1) \leftarrow \InitInfo()$}\;\label{alg:line:init}
    \For{\label{alg.line1}each round t}
    {
        If a job $j$ is {submitted}, choose launch plan $L_j$ as per \eqref{eq:choose-L}
        \; \label{alg:pick}
        {\For{each active job $j$ that arrives at round t}{
            schedule $j$ to start at time {$L_j(t,\sigma_j)$}\;
        }}
        \While{the current committed load exceeds $C$}{
            Evict/cancel the most-recently-submitted active job\;
        }
        Start executing each active job $j$ scheduled to start at $t$\;
        \For{each job $j$ successfully completed at round t}{
            Charge agent $j$ a payment of $\pay^{\birth_j}(L(\arr_j),\dem,\dur_j)$\;
        }
%        Update announced info
        %$\pay^{t+1}$ and $\evictapx^{t+1}$
%        \ASedit{$(\pay^{t+1}, \evictapx^{t+1})$: call $\UpdateInfo()$}\;            \label{alg:line:update}
        {Update: $(\pay^{t+1}, \evictapx^{t+1}) \leftarrow \UpdateInfo()$}\;            \label{alg:line:update}
	}
\end{algorithm}

\xhdr{Announced info.}
At each round $t$, the scheduler {updates and} announces two pieces of information for jobs that are submitted (and born) in this round: a price menu $\pay^t$ and \emph{estimated} failure  probabilities $\estP^t$. Functions $\pay^t$ and $\estP^t$ are computed without observing the \SoWs for these jobs.
%Suppose that after $\pay^t$ and  $\estP^t$ are announced, a job has reported concurrency demand $\dem_j$, {realized} duration $\dur$, and starts executing in some round $t'\geq t$.
{The meaning of $\pay^t$ and $\estP^t$ is as follows. Suppose job $j$ submitted at time $t$} has reported concurrency demand $\dem_j$, realized duration $\dur$, and starts executing in some round $t'\geq t$. The scheduler announces the price $\pay^t(t',\dem,\dur)$ that would be paid if such job successfully completes, and an estimated probability that it would not complete, denoted $\estP^t(t',\dem,\dur)$. These are announced for all relevant  $(t',\dem,\dur)$ triples, \ie for all rounds
    $t'\in [t,t+\maxS]$,
demands
    $\dem\in[\maxC]$,
and durations
    $\dur\in [\maxD]$.
By convention, we set $\pay^t(\infty,\dem,\dur)=0$.
%\iffalse
%Prices and estimated failure probabilities may vary with $t$.  The price $\pay^t(t',\dem,\dur)$ does not %decrease with {announce time} $t$, and costs are non-decreasing in both duration and concurrency: %$\pay^t(t',\dem,\dur) \leq \pay^t(t',\dem',\dur')$
%     for any
%        $\dem \leq \dem',\, \dur \leq \dur'$.
%\fi
Prices may vary with $t$, within two invariants:
\begin{OneLiners}
    \item
    %Costs never decrease over time:\\
    $\pay^t(t',\dem,\dur)$ does not decrease with {announce time} $t$.

    \item Costs are non-decreasing in both duration and concurrency:
        $\pay^t(t',\dem,\dur) \leq \pay^t(t',\dem',\dur')$
     for any
        $\dem \leq \dem',\, \dur \leq \dur'$.
\end{OneLiners}
Likewise, estimated failure probabilities
may vary over time, but are
non-increasing in both duration and concurrency:
$\estP^t(t',\dem,\dur) \leq \estP^t(t',\dem',\dur')$
for all $\dem \leq \dem',\, \dur \leq \dur'$.
\AScomment{inlined}

%\[
%\estP^t(t',\dem,\dur) \leq \estP^t(t',\dem',\dur')
%\quad
%\ \forall\dem \leq \dem',\, \dur \leq \dur'.
%\]
%\[\evictapx^t(t',\dem,\dur) \leq \evictapx^t(t',\dem',\dur') \quad \forall \dem \leq \dem', \dur \leq \dur'.\]

%obey two invariants:
%\begin{OneLiners}
%    \item Wider jobs are at greater eviction risk:
%    $\mE_t(\arr,c,\dur)$ does not decrease with $c$.
%    \item Longer jobs are at greater eviction risk:
%        $\mE_t(\arr,c,\dur)$ does not decrease with $\dur$
%\end{OneLiners}

An instantiation of Algorithm~\ref{alg.lifo} should {implement $\InitInfo()$ and $\UpdateInfo()$}. The rest of the algorithm is then fixed.

\xhdr{Launch plans.}
At each round $t$, upon receiving the \SoW for a given job $j$, the scheduler computes the {\em launch plan} $L_j$ for this job, which maps every possible arrival time $\arr_j$ {and signal $\sigma_j$ (if any)} to the start time of the execution. {The launch plan may decide to \emph{not} execute the job for some arrival times $\arr_j$; we denote this {$L_j(\arr_j,\sigma_j)=\infty$}.} The launch plan is binding: job $j$
%that arrives at time $\arr_j$
must start executing at time {$L_j(\arr_j,\sigma_j)$}, unless it is \emph{cancelled} beforehand (as explained below).

The choice of a launch plan, described below, is crucial to ensure incentives. For a given launch plan $L$ and a job whose true \SoW is $(V,\dem,P)$, we define the \emph{estimated utility}
    $\estU_t(L\mid V,\dem,P)$
as the agent's expected utility under the announced prices $\pay^{t}$, assuming that the estimated failure probabilities are correct. In a formula,
\begin{align}\label{eq:util-L}
&\estU_t(L\mid V,\dem,P) \\
&\quad={\E_{(\arr,\dur) \sim P}
    \sbr{ \rbr{1-\estP^{t}(t_{a,\sigma},\dem,\dur)}\cdot
        \rbr{V(t_{a,\sigma}+\dur) - \pay^{t}(t_{a,\sigma},\dem,\dur)}},}
%         \rbr{1-\evictapx^{t}(L(\arr),\dem,\dur)}},
         \nonumber
\end{align}
where
%$\sigma = \sigma(a,d)$ and {$t_{a,\sigma} = L(a,\sigma)$}.
{$t_{a,\sigma} = L(a,\sigma(a,d))$}.
We choose a launch plan 
\begin{align}\label{eq:choose-L}
L_j \leftarrow \argmax_{\text{launch plans $L$}}
    {\estU_t(L\mid V_j,\dem_j,P_j)}.
    %\estU_t(L \mid \SoW_j).
\end{align}
so as to maximize the estimated utility given the reported \SoW.

%After a launch plan is chosen, the scheduler updates the prices and failure probability estimates for future jobs. The updates can be randomized.  An instance of the scheduling framework, such as the one we describe in Section~\ref{sec:reducing}, will specify how this update should be performed.

\begin{remark}\label{rem:multiple-births}
{For convenience we described Algorithm~\ref{alg.lifo} under the assumption that at most one job is submitted each round.  This can be relaxed: if $r$ SoWs are simultaneously submitted at time $t$, we would choose a launch plan for each job sequentially (in any order), update the announced info after each job, and then move to
schedule arriving jobs %clearing over-committed resources
(line %\ref{line:over}
4) after all $r$ jobs have been handled.
%
%by allowing submission times to be continuous and processing jobs (by choosing a launch plan and updating $\pay^t$ and $\evictapx^t$) as they are submitted. \mbc{I find this not very clear. I think it is easier to say that if $r$ SoWs are simultaneously submitted at time $t$, we can treat time $t$ as $r$ separate times, each with one job submitted, handle the selection of the launch plan for each sequentially (in any order) and then move to clearing over-committed resources (line 4) after we are done. }
%We would therefore interpret $t$ in $\pay^t$ and $\evictapx^t$ as potentially fractional, as announced values could change multiple times between rounds. Note that we would still require that start times and durations are aligned to discrete timesteps.
%For brevity we will refrain from discussing this further in the sequel.
}
%; it is only the announce time $t$ in $\pi^t$ and $\evictapx^t$ that is
%\ASedit{TODO: how to extend to multiple births per round}.
\end{remark}

%\AScomment{completely rewrote the next subsection}

\xhdr{Cancellations and evictions.}
The scheduler can \emph{cancel} a job that has not yet started executing, or \emph{evict} a job that has.  We never restart an evicted or canceled job.
%Both are irrevocable: the job is permanently removed from the system.
A job is called \emph{active} at a given point in time if it has  been submitted, but has not yet been completed, cancelled, or evicted.\footnote{{As a convention, after the last round ($t=T$) all active jobs are cancelled or evicted.}}
%Then, conveniently, every job either completes or is cancelled/evicted.}}
{We say that an active job $j$ is \emph{scheduled to start} at round $t$ if it has
%already
arrived and {$t=L_j(\arr_j,\sigma_j)$.}} The \emph{current committed load} is the total concurrency demand, $\sum_j \dem_j$, of all active jobs $j$ that are executing or scheduled to start in the current round.

%\AScomment{subtle/annoying detail: "active" and "current committed load" are defined "for a given point in the execution of the mechanism", not "for a given round". Because we may cancel jobs in the middle of the round.}

If the current committed load is above the total supply $C$, the scheduler evicts or cancels active jobs in LIFO order of birth time (most recently born first) until the current load is at most $C$. A job is charged zero payment if it is evicted or cancelled.

\begin{remark}\label{remark:LIFO-all}
The LIFO order is over \emph{all} active jobs, including jobs that are not scheduled to run in the current round.  While this feature is not necessary to address an overbooking failure, it is crucial to our analysis, as explained below.%
\footnote{In practice, one might decide \emph{not} to cancel jobs that have not yet started executing. While such a modification would perturb customer incentives, in a real system it might be acceptable as finding beneficial misreporting might be challenging.}
\end{remark}

%\AScomment{it would be nice to unite this Remark with the next one, but they are making slightly different points which belong to their respective places ...}

%This execution strategy is described as Algorithm~\ref{alg.lifo}.
%An important property of this eviction/cancellation rule is that

We observe that the eviction/cancellation probabilities for a given job are determined at birth/submission time.
%The exact statement that we need is somewhat subtle. For brevity,
Formally, let $\term{FAIL}_j$ denote the event that job $j$ does not successfully complete, and let $\hist^t$ denote the full history of events observed by the algorithm up to (and not including) round $t$ (including all \SoWs submitted, launch plans chosen, realized arrivals, job completions, and evictions/cancelations).
Also, let $\SoW_{[t,t']}$ denote the collection of \SoWs for all jobs submitted in the time interval $[t,t']$.

%\AScomment{we've been informal about this statement. I found it very tricky to formulate. Therefore, I think we should spell it out. However, please feel free to simplify if you can.}

\iffalse
\begin{lemma}\label{lm:fail-prob-nice}
Consider some round $t$, and fix
%tuple $(t, t', \dem, \dur)$.
start-time $t'\geq t$, concurrency demand $\dem\in [\maxC]$, and duration $\dur\in[\maxD]$.
Suppose job $j$ is submitted in round $t$, with $\dem_d=\dem$ and $\dur\in \term{support}(\dur_j)$. Pick any launch plan $L$ consistent with $t'$ and $\dur$, in the sense that
    $ \Pr\sbr{ L(\arr_j)=t' \text{ and } \dur_j=\dur} >0$,
and suppose this launch plan is chosen for job $j$, instead of $L_j$. Then
\begin{align}\label{eq:fail-prob-nice}
 \Pr\sbr{\term{FAIL}_j \mid L(\arr_j) = t',\; \dur_j = \dur,\; \SoW_{[1,T]} }
\end{align}
is determined by $\SoW_{[1,t-1]}$ (the past SoWs) and
$(t,t',\dem,\dur)$. Given these quantities, \eqref{eq:fail-prob-nice} does \emph{not} depend on the launch plan $L$ and the reported \SoWs of this and future jobs.
\end{lemma}
\fi

\begin{lemma}\label{lm:fail-prob-nice}
Consider some round $t$ and fix tuple $(t, t', \dem, \dur)$.  Suppose a job $j$ is submitted in round $t$ with $\dem_j=\dem$, and suppose there is a launch plan $L$ such that
    $ \Pr\sbr{ {L(\arr_j,\sigma_j)}=t' \text{ and } \dur_j=\dur} >0$.
Then if launch plan $L$ were chosen for $j$ (i.e., ignoring \eqref{eq:choose-L}), then
%Consider a job $j$ submitted in round $t$ with $\dem_j = c$.  Consider any launch plan $L$ feasible for job $j$ and suppose that, for some tuple $(t',a,d)$, $\Pr[ \arr_j = a,\; L(\arr_j) = t',\; \dur_j = d] > 0$. Then
\begin{align}\label{eq:fail-prob-nice}
 \Pr\sbr{\term{FAIL}_j \mid \hist^t,\; L(\arr_j,\sigma_j) = t',\; \dur_j = \dur,\; \SoW_{[1,T]} }
\end{align}
is determined by  %$\SoW_{[1,t-1]}$ (the previously-submitted \SoWs) and
$\hist^t$ and $(t,t',\dem,\dur)$ (and independent of $\SoW_{[t,T]}$).
\end{lemma}

Lemma~\ref{lm:fail-prob-nice} follows immediately from the LIFO ordering: a job $j$ will be evicted/canceled only if the committed load exceeds $C$ even after removing all subsequently-submitted jobs, and this depends only on the launch plans and realizations of previously-submitted jobs.  Given Lemma~\ref{lm:fail-prob-nice}, we can denote \eqref{eq:fail-prob-nice} by $\evict^t(t',\dem,\dur)$, and call it the (true) \emph{failure probability}. We will interpret $\estP^t(t',\dem, \dur)$ as an approximation of $\evict^t(t',\dem,\dur)$.

%\begin{remark}
%The eviction/cancellation rule is important both for incentives and for calculating optimal schedules. In particular, the information needed to compute $\evict^t(t',\dem,\dur)$ is fully available to the scheduler at time $t$. Moreover, no \SoWs submitted in the future will impact the scheduler's ability to execute this job's launch plan.
%\end{remark}

While \eqref{eq:fail-prob-nice} can, in principle, be computed exactly, such computation may be infeasible in practice. We only require the estimates to be approximately correct: we bound the error by some $\mu>0$, and bound the possible gains from untruthful reporting in terms of $\mu$. Specifically, we assume
\ASedit{that, taking the expectation over $\estP^t$,}
%
%Note that we do not require that $\evictapx^t = \evict^t$.
%
%\mbe{It will not be feasible to have the estimated success probability be perfectly accurate (e.g., due to computation constraints), yet, as long as they are close enough to the true probabilities, the possible gain from being untruthful will be bounded.
%OLD:
%While we do not require     $\estP^t = \evict^t$
%%$\evictapx^t = \evict^t$,
%we need the estimates  to be close to the true eviction probabilities:
%We will use $\mu>0$ to denote an upper bound on the error:}
%
\begin{align}\label{eq:eviction-estimates}
\E\sbr{|\estP^t(t',\dem, \dur) - \evict^t(t',\dem,\dur)|} < \mu
\quad(\forall \dem,\dur,t'\geq t).
\end{align}
%where the expectation is over any randomness in $\estP^t$.

% The uncertainties around job arrivals and durations mean that the set of launch plans implies a distribution over the number of busy compute nodes at each time. The entropy of that distribution increases as new \SoWs arrive, and decreases as jobs actually arrive and ultimately complete running.

%Assuming that the eviction probability distribution $\mE_t(t',c, d)$ is accurate (\ie, is provided by an oracle), the resulting framework is incentive compatible for customers in the following sense.

%\AScomment{No need to assume that eviction probabilities are accurate -- we've assumed it already!}

\xhdr{Incentives.}
Without detailing how the prices are selected and how the estimated success probabilities are computed,
%when agent are willing to tolerate some small utility loss then
we can already guarantee approximate truthfulness. Essentially, this is because launch plans optimize agents' expected utility with respect to the approximate failure probabilities.

\begin{theorem}\label{thm:incentives}
%If $E[|\evictapx^t(t',c, d) - \evict^t(t',c,d)|] < \mu$ for all $t$ (where the expectation is over any randomness in $\evictapx)$, then
Algorithm~\ref{alg.lifo} is $(2\mu \maxV)$-truthful,
where $\mu$ bounds the success probability estimation error (as in \eqref{eq:eviction-estimates}) and positive job values are normalized to lie in $[1,\maxV]$.
\end{theorem}

\vspace{-1mm}\xhdr{Computation.}
%We now turn to computational considerations.
To compute the optimal $L_j$ in \eqref{eq:choose-L}
%Line~\ref{alg1.line3},
one can separately optimize $L_j(\arr_j,\sigma_j)$ for each potential arrival time $\arr_j$ and signal $\sigma_j$. This optimization can be done by enumerating over each $(\arr,\dur)$ in the support of $P_j$ and each potential start time.  One can therefore compute the optimal launch plan in time $O(\maxS\cdot |\term{support}(P_j)| )$.
\section{Reduction Approach}
\label{sec:reducing}

We reduce the original problem (henceforth called \MainProb), to its relaxation, \RelaxProb. The latter is a different but related scheduling problem, where job requirements are fractional rather than uncertain, and the load corresponds to the expected load in \MainProb.
%We transform any instance of \MainProb to some instance of \RelaxProb, and solve the former given an algorithm for the latter.
{Our reduction takes an algorithm for \RelaxProb,
in which over-commitment is never an issue, and use it to solve \MainProb where the system might get saturated in some realizations.}
%\mbc{maybe say a word about the main challenge - why isn't this "simply taking expectation"}
In Section~\ref{sec:oracle} we complete this approach by adapting known online resource allocation techniques to solve \RelaxProb.

%To design algorithms in our framework, we will show how to reduce to a scheduling problem without job uncertainty.  We will then argue that a particular natural form of online algorithm for the relaxed problems translates into an algorithm for the original problem.

\subsection{The Relaxed Problem}

%We will refer to the scheduling problem described in Section~\ref{sec:prelim} as online scheduling with job uncertainty (OSJU).
% The \emph{fractional scheduling problem} is similar to OSJU

\RelaxProb is similar to \MainProb, with these changes:
%\begin{enumerate}

\fakeItem Each job $j$ is characterized by a \emph{fractional \SoW}, which contains value function $V_j(\cdot)\geq 0$ and concurrency demand $\dem_j$ as before, but distribution $P_j$ is replaced with $k_j$ \emph{tasks} $\tau_1 \LDOTS \tau_{k_j}$ and weights $\lambda_1 \LDOTS \lambda_{k_j} > 0$ with $\sum_i \lambda_i = 1$. Each task $\tau_i$ is specified by an arrival time and duration $(\arr_{ij}, \dur_{ij})$.

\fakeItem An allocation to job $j$ assigns to each of its tasks $\tau_i$ either no resources, or $\lambda_i\dem_j$ resource units for $\dur_{ij}$ consecutive timesteps starting no earlier than $\arr_{ij}$. Note that $\lambda_i \dem_j$ might be fractional. Write $x_{ijt} \geq 0$ for the amount of resources allocated to task $\tau_i$ at time $t$, and write $\finish_{ij}$ for the completion time of this task, or $\finish_{ij}=\infty$ if it is not completed. The allocation for a single task is called an \emph{interval allocation} and denoted
        $x_{ij} = \rbr{x_{ijt}:\, t\in[T]}$.
    The \emph{aggregate allocation} for job $j$ denoted
        $x_j = \rbr{x_{ijt}:\, \text{tasks $i$, rounds $t$}}$.

\fakeItem The value of interval allocation $x_{ij}$ is $\lambda_i V_j(\finish_{ij})$.
The value of the aggregate allocation $x_j$ is
    $\tilde{V}_j(x_j) = \sum_i \lambda_i V_j(\finish_{ij})$.

\fakeItem When a fractional \SoW for a given job is submitted, its allocation must be irrevocably decided right away. Tasks cannot be evicted, preempted or cancelled afterwards.

\fakeItem The total allocation to all jobs $j$ and tasks $i$ at any time $t$ cannot exceed $C$, \ie $\sum_{i,j} x_{ijt} \leq C$.
%\bledit{The set of allowable aggregate allocations may be subject to other feasibility constraints as well.}
%\end{enumerate}
\vspace{1mm}

%\paragraph{A Correspondence Between Problems}

As before, job's birth times and fractional \SoWs comprise a \emph{birth sequence}, which is chosen ahead of time from some distribution over birth sequences. This distribution constitutes a problem instance.

Given an instance $\mI$ of \MainProb, we construct an instance of \RelaxProb, denoted $\Relax(\mI)$, in a fairly natural way.
For each
    $\SoW_j = (V_j, \dem_j, P_j)$
in \MainProb, the corresponding fractional \SoW has the same $V_j$ and $\dem_j$, and tasks $\tau_i = (a_{ij}, d_{ij})$ for each $(a_{ij}, d_{ij})$ in the support of $P_j$, with weights
    $\lambda_i = P_j\sbr{(a_{ij}, d_{ij})}$.
We denote this fractional \SoW as $\Relax(\SoW_j)$.

%    $\lambda_i = \Pr_{P_j}[(a_{ij}, d_{ij})]$.

Any launch plan $L_j$ for job $j$ in \MainProb assigns to each $(\arr_{ij}, \dur_{ij})$ an interval allocation of $\dem_j$ resources for $\dur_{ij}$ rounds starting at {$L_j(\arr_{ij},\sigma(\arr_{ij},\dur_{ij}))$}.  This corresponds to an interval allocation $x_{ij}$ to each task $\tau_i$ in the fractional scheduling problem, in which the resources allocated each round are scaled by $\lambda_i$.  We will write $x_j(L_j)$ for the aggregate allocation (for all tasks). Note then that
    $\sum_i x_{ijt}(L_j)$
is the expected usage of resources at time $t$ under launch plan $L_j$, with respect to probability distribution $P_j$.
%\bledit{In our algorithms and benchmarks we will restrict attention to aggregate allocations that correspond to launch plans.}

%implies an expected commitment of resources $x_{jt}$ at each time $t$, where the expectation is over the realization of $(a_t, d_t) \sim P_j$.  Write $x(L_j)$ for this expected commitment of resources by launch plan $L_j$.  We can therefore define a corresponding job $j$ for the fractional scheduling problem, with valuation $\tilde{v}_j$ defined so that $\tilde{v}_j(x_j)$ is the maximum, over launch plans consistent with expected commitments $x_j$, of the expected value obtained from $L_j$.  That is,
%\[ \tilde{v}_j(x_j) = \max_{L_j \colon x(L_j) = x_j}\{E_{(a_j, d_j) \sim P_j}[V_j(a_t, L_j(a_t)+d_t]\}. \]
%We note that in the resulting fractional scheduling problem, our capacity constraint is that $\sum_j x_{jt} \leq C$ for each $t$. That is, it requires only that the capacity constraint $C$ is satisfied in expectation over the realization of job requirements.  On the other hand, OSJU requires that the capacity constraints hold for every realization, not just for the expected commitments, but allows jobs to be evicted and/or canceled.

\xhdr{A class of algorithms.}
%\mbe{In order to later be able to reduce our problem with uncertainty to \RelaxProb, we need algorithms for \RelaxProb that have the following special structure.}
{Our reduction requires algorithms for \RelaxProb with the following special structure.}

{First,} an algorithm maintains a price function $\altpay$, a.k.a. a \emph{menu}, that assigns a non-negative price to any interval allocation $x$. The menu can change over time as the algorithm progresses, so we write $\altpay^t(x)$ for the price at time $t$. When a given job $j$ is submitted at time $t$, the algorithm optimizes the allocation $x_j$  according $\pi^t$:
\begin{align}\label{eq:frac-alloc-menu}
x_j \in \argmax_{\text{aggregate allocations $x_j$}}
    \tilde{V}_j(x_j) - \altpay^t(x_j),
\end{align}
where the total job price is %of the job is
    $\altpay^t(x_j) = \sum_{\text{tasks $i$}} \altpay^t(x_{ij})$.
Such algorithms are called \emph{menu-based}.\footnote{{Note that the $\argmax$ in \eqref{eq:frac-alloc-menu} is over \emph{all} aggregate allocations $x_j$ that correspond to launch plans, some of which may violate capacity constraints. To be feasible, the menu must ensure that the output of the $\argmax$ stays within the constraints.}}  We write $\altpay^t(t',\dem,\dur)$ for the price of allocating $\dem$ resources for $\dur$ steps starting at time $t'$.

%For an allocation $x_j$ to a job $j$ (which recall is a collection of interval allocations, one per task), we write $\altpay^t(x_j) = \sum_i \altpay^t(x_{ij})$ for the total price of $x_j$.  Then a menu-based algorithm allocates to each job $j$ born at time $t$ an assignment $x_j \in \argmax_{x_j} \{ \tilde{V}_j(x_j) - \altpay^t(x_j) \}.$

%The menu-based algorithms we consider will be competitive against a strong benchmark: the optimal solution to a relaxed offline allocation problem.

{Second, the algorithm is measured with respect to the following strong benchmark:} for any subset $N$ of jobs, $\opt_N$ is the offline optimal welfare attainable over jobs $N$ in a randomized version of \RelaxProb, where the choice of allocation to each task can be randomized and the supply constraints need only bind in expectation over this randomization. We write
    $\opt = \opt_{\{\text{all jobs}\}}$.

{Third, an algorithm should support the following partially adversarial scenario.}
When a job is born, an adversary can arbitrarily break ties in the choice rule \eqref{eq:frac-alloc-menu}. Moreover, the adversary can bypass \eqref{eq:frac-alloc-menu}, and instead choose any allocation $x_j$ such that $\tilde{V}_j(x_j) \geq \altpay^t(x_j)$.  {When/if this happens, the algorithm should observe the new $x_j$ and continue. Call such algorithms \emph{receptive}.} %after this adversarial choice,
%and in particular can still update prices.
 %given the adversary's choice.
The algorithm does not need to compete with $\opt$. Instead, it
only needs to compete with $\opt_N$, where $N$ is the set of jobs whose allocation satisfies \eqref{eq:frac-alloc-menu} (\ie is not switched by the adversary).
%Write $A$ for the set of jobs scheduled by the adversary, and $N$ for the remaining jobs.
%We then say that
The algorithm is \emph{robustly $\alpha$-competitive}, $\alpha\geq 1$ on a given problem instance if for any adversary, the total value generated by the algorithm,
    $\sum_j \tilde{v}_j(x_j)$
%of the allocation generated by the algorithm
(including jobs scheduled by the adversary)
is at least $\frac{1}{\alpha}\; \opt_N$.

{In summary, an algorithm for \RelaxProb that is menu-based and receptive uses the following protocol in each round $t$:
\begin{OneLiners}
\item[1.] if a new job $j$ arrives, choose an allocation as per \eqref{eq:frac-alloc-menu},
\item[2.] replace with the adversarial allocation if applicable,
\item[3.] update price menu $\altpay^t$.
\end{OneLiners}
}

%is within an $\alpha$ fraction of the optimal relaxed allocation of only the jobs in $N$.  That is, for any choices of the adversary, if $x$ is the allocation generated by the algorithm, we have that
%\[ \sum_j \tilde{v}_j(x_j) \geq \frac{1}{\alpha} \opt_N \]

\subsection{Reduction to the Relaxed Problem}
\label{sec:relaxed.reduction}

We instantiate Algorithm~\ref{alg.lifo} using a menu-based, {receptive} algorithm $\alg$ for \RelaxProb. This instantiation (Algorithm~\ref{alg.reduction}) is competitive for any instance $\mI$ of \MainProb as long as $\alg$ is robustly-competitive on the relaxed instance $\Relax(\mI)$.

\begin{theorem}\label{thm:reduction}
Fix $\epsilon > 0$ such that system's capacity $C$ exceeds
    $\Omega( \maxC\,\epsilon^{-2} \log( \epsilon^{-1} + \maxS ) )$.
Consider an instance $\mI$ of \MainProb and a menu-based, {receptive} algorithm $\alg$ for \RelaxProb. Suppose $\alg$ is robustly $\alpha$-competitive for the relaxed problem instance $\Relax(\mI)$ and some $\alpha\geq 1$. Then Algorithm~\ref{alg.reduction} with parameter $\epsilon$ is {$O(\epsilon\maxV)$-truthful and} $\alpha(1+O(\epsilon))$-competitive for the original problem instance $\mI$. {The per-round running time is $O(\epsilon^{-2} \log (\epsilon^{-1} \maxD \maxS)))$ plus the per-round running time of $\alg$.}
\end{theorem}

\begin{algorithm}[t]
\alginit
\caption{Algorithm~\ref{alg.lifo} via \RelaxProb}
\label{alg.reduction}
\setcounter{AlgoLine}{0}
\SetKwInOut{Parameters}{Parameters}
\SetKwInOut{Require}{Require}
\Parameters{$\epsilon > 0$, menu-based, {receptive} algorithm $\alg$}
\tcp{$\alg$ uses supply constraint $\algC = C\cdot (1-\epsilon/10)$}
\Require{Oracle for failure probabilities $\evict^t$}
\vspace{2mm}
\tcp{$\InitInfo()$ for Line~\ref{alg:line:init} of Algorithm~\ref{alg.lifo}}
\vspace{-1mm}
Initialize prices $\pay^1$ same way as $\alg$; set $\estP^1\leftarrow 0$\;
\vspace{2mm}
\tcp{$\UpdateInfo()$ for  Line~\ref{alg:line:update} of Algorithm~\ref{alg.lifo}}
\If{job $j$ is submitted in round $t$}{
    Report job $\Relax(\SoW_j)$ to $\alg$\;
    {tell $\alg$: for job $j$, use allocation $x_j(L_j)$}\;
    \label{line:reduction-tell}
    \tcp{Claim~\ref{cl:reduction-tell}: $x_j(L_j)$ "typically" maximizes \eqref{eq:frac-alloc-menu}}
%    \eIf{
%        $\estP^t(t',c,d)=0$ for all pairs $(a,d)\in \term{support}(P_j)$ and all  times $t'\in [a,a+\maxA]$
%    }{
%        \ASedit{tell $\alg$: for job $j$, break tie in \eqref{eq:frac-alloc-menu} in favor of  $x_j(L_j)$}
%        \label{line:reduction-main}
%    }{
%        tell $\alg$: for job $j$, use allocation $x_j(L_j)$\;
%        \label{line:reduction-exception}
%    }
}
Update menu $\pay^{t+1}$ to the updated price menu from $\alg$\;
Update $\evict^{t+1}$ via the oracle and
    $\estP^{t+1} \leftarrow f(\evict^{t+1})$,
where
\begin{align*}%\label{eq:thresh}
\forall q\geq 0\;\;
\text{$f(q) = q$ if $q>\epsilon/10$, and $0$ otherwise}.
\end{align*}
\end{algorithm}

Our reduction proceeds as follows. Formally, in Algorithm~\ref{alg.reduction} we fill out the two unspecified steps
in Algorithm~\ref{alg.lifo}, {$\InitInfo()$ and $\UpdateInfo()$.}
%how to initialize $\pay^1, \estP^1$ in Line~\ref{alg:line:init}, and how to update $\pay^{t+1}, \estP^{t+1}$ in Line~\ref{alg:line:update}.
Substantively, we simulate a run of $\alg$ on the relaxed problem instance $\Relax(\mI)$. Whenever a new job $j$ is submitted, we report its fractional version $\Relax(\SoW_j)$ to $\alg$, and use the price menu previously computed by $\alg$ to optimize the launch plan $L_j$.
Then we force $\alg$ to follow the same launch plan for this job: namely, use  aggregate allocation $x_j(L_j)$ for job $j$. If all relevant estimates $\estP^t$ are zero, this choice just breaks ties in \eqref{eq:frac-alloc-menu}:

\begin{claim}\label{cl:reduction-tell}
In Line~\ref{line:reduction-tell} of Algorithm~\ref{alg.reduction}, suppose $\estP^t(t',c,d)=0$ for all pairs $(a,d)\in \term{support}(P_j)$ and all  times $t'\in [t,t+\maxS]$. Then aggregate allocation $x_j(L_j)$ maximizes \eqref{eq:frac-alloc-menu}.
\end{claim}

We posit oracle access to the (true) failure probabilities $\evict^t$. The simplest version is that the oracle returns exact probabilities. By a slight abuse of notation, we allow the oracle to be $\epsilon_0$-approximate with probability at least $1-\delta_0$, for some $\epsilon_0 = \delta_0 = \Theta(\epsilon)$. In Appendix~\ref{sec:approximate}, we provide an efficient procedure to compute such estimates.

Once we get $\evict^t$ from the oracle, we compute the estimates $\estP^t$ in a somewhat non-intuitive way: we zero out all estimates smaller than a given threshold. Put differently, we ignore failure probabilities if they are sufficiently small. This choice is crucial to ``inherit" the performance guarantee of $\alg$, as we show in the next section.
%If we zero out \emph{all} relevant estimates $\estP^t$, then our launch plans would always be consistent with the allocations chosen by $\alg$; however, their performance may be off because of failure events that we didn't account for. \ASdelete{With the threshold version, we need to bound the loss of $\alg$ due to adversarially chosen jobs, and the loss in our launch plans due to the thresholding.}

\subsection{Proof of Theorem~\ref{thm:reduction}}

%\subsection{Analysis of Algorithm~\ref{alg.reduction}}
\label{sec:analysis}

%\blcomment{Got to here.}

%This section is dedicated to the proof of Theorem~\ref{thm:reduction}.
%we bound the competitive ratio of Algorithm~\ref{alg.reduction}.
%Our strategy is to
{We argue} that each job $j$ will face low failure probabilities, in the sense of Claim~\ref{cl:reduction-tell}, with high probability.
%This will have two implications: that
{Then (i)} the total value obtained by Algorithm~\ref{alg.reduction} is close to the simulated value obtained in our simulation of $\alg$, and {(ii)} the simulated value of $\alg$ is large compared to $\opt$.  We now formalize this intuition.
%(and not just $\opt_N$, where $N$ is the set of  restricted %, then compare this with the total value obtained by Algorithm~\ref{alg.reduction}.
%with the expected value of $\opt_N$, the optimal relaxed allocation of jobs scheduled non-adversarially in the simulated problem.
%\mbc{the above in somewhat unclear. Maybe replace by (if I did not mess it): "Our strategy is to argue that in the simulation, every job has high probability of being executed as $\alg$ would execute its fractional version,
%or essentially as it would when there is no capacity constraint."}

Fix some birth sequence, and for convenience write $V^{\alg}$ for the total simulated value obtained by $\alg$ in Algorithm~\ref{alg.reduction}.  Let $N$ be the set of jobs $j$ for which $x_j(L_j)$, from Line~\ref{line:reduction-tell} of Algorithm~\ref{alg.reduction}, maximizes \eqref{eq:frac-alloc-menu}.  That is, $N$ is the set of jobs whose allocations were not adversarially switched in our simulation of $\alg$.  Then since $\alg$ is robustly $\alpha$-competitive, we know that $V^{\alg} \geq \frac{1}{\alpha} \opt_N$.

Since we actually want to compare $V^{\alg}$ with $\opt$, we need to show that $\opt_N$ is close to $\opt$.  By Claim~\ref{cl:reduction-tell}, we will have $j \in N$ whenever all eviction probabilities are sufficiently small for job $j$.  So our goal is to establish that each job is very likely to face very low failure probabilities in Algorithm~\ref{alg.reduction}.  This is the most technical step in the proof.
Intuitively, since $\alg$ constructs allocations subject to a reduced supply constraint $\algC$, concentration bounds suggest that it's exponentially unlikely that total realized usage will exceed $C$ in any given round.
However, there is correlation between failure probabilities in different rounds.
%If the failure probability becomes high at some time $t$, then this will influence the choice of allocation made by the algorithm, which in principle could increase the failure probability at a later time $t' > t$.
One might therefore worry that even if it takes exponential time for a first eviction to occur, evictions would become more common thereafter.  We must therefore bound the impact of correlation across time.
%, and for this we use the fact that each job has only a finite window of time in which it can execute.
This is accomplished by the following lemma {(proved in the Appendix)}.
%\bledit{Recall that $\maxS $ is an upper bound on the difference between a job's submission time $B_j$ and the latest start time such that the resulting job's value can be non-zero.}

\begin{lemma}\label{lem:eviction}
Fix any sequence of job birth times and \SoWs, and choose any $\lambda > 0$ and $\delta > 0$.  If $\algC < (1-\delta)(C - \maxC)$, then
%then for any times $t$ and $t' > t$, and any $c \leq \maxC$ and $d \leq \maxD $,
\begin{align}\label{eq:lem:eviction}
\Pr\sbr{\evict^{t}(t',\dem,\dur) > \lambda} <
(\maxS)^2\cdot \lambda^{-1} \cdot
 e^{-\Omega\rbr{(C/\maxC-1)\cdot \delta^2/(1+\delta)}}
%e^{-\Omega(\frac{\delta^2}{1+\delta}\cdot\frac{(C-\maxC)}{\maxC} )}
%\quad \forall t, t'>t, \dem \leq \maxC, \dur \leq \maxD
\end{align}
for all $t$, $t'>t$, $\dem \leq \maxC$, $\dur \leq \maxD$, 
where $\Pr[]$ the arrival times and durations for all jobs.
\end{lemma}

{Let $\gamma$ denote the right-hand side in \eqref{eq:lem:eviction}.}
%Write $\gamma = \frac{(\maxS)^2}{\lambda} \cdot e^{-\Omega(\frac{\delta^2}{1+\delta}\cdot\frac{(C-\maxC)}{\maxC} )}$ for the bound in Lemma~\ref{lem:eviction}.
Assume for now that $\gamma = O(\epsilon)$.  If we set $\lambda = \epsilon/10$ (the threshold for $\estP$ in Algorithm~\ref{alg.reduction}), Lemma~\ref{lem:eviction} implies that if the oracle for failure probabilities is perfectly accurate, each job will lie in $N$ with probability at least $(1-\gamma)$.
%This lets us conclude that $\opt_N \geq (1-\gamma)\opt$, and hence $V^{\alg} \geq (1-\gamma)\frac{1}{\alpha}\opt$.
If instead our failure probability oracle is only $\epsilon_0$-approximate with probability at least $1-\delta_0$, where $\epsilon_0 = \delta_0 = O(\epsilon)$, then we would instead set $\lambda = \epsilon/10 + \epsilon_0$ to conclude that each job will lie in $N$ with probability at least $(1-\gamma-\delta_0)$. This lets us conclude that $\opt_N \geq (1-\gamma-\delta_0)\opt$, and hence $V^{\alg} \geq (1-O(\epsilon))\frac{1}{\alpha}\opt$.

%We note that Lemma~\ref{lem:eviction} holds for any sequence of job birth times and \SoWs.  The failure probability $\evict^t$ at time $t$ depends on the allocation choices made by Algorithm~\ref{alg.reduction} up to time $t$, which in turn depends on the realizations of $(\arr_j, \dur_j) \sim P_j$ for each job $j$.

The next step is to compare the total value obtained by Algorithm~\ref{alg.reduction} to the simulated value $V^{\alg}$.  The difference between these quantities is that jobs may be evicted in \MainProb, in which case they contribute to the simulated value but not the true realized value.  But jobs in $N$ are evicted with probability at most $\epsilon/10 + \epsilon_0 + \delta_0$, by definition of $N$ and the estimation guarantees of our oracle, and by Lemma~\ref{lem:eviction} each job lies outside $N$ with probability at most $\gamma$.  So each job is evicted with probability at most $(\gamma + \epsilon/10 + \epsilon_0 + \delta_0) = O(\epsilon)$. The total value obtained by Algorithm~\ref{alg.reduction} is therefore at least $(1 - O(\epsilon)) V^{\alg} \geq (1-O(\epsilon))\frac{1}{\alpha}\opt$.

Finally, we bound the effect of reducing the supply to $\algC = C\cdot (1-\epsilon/10)$ in our simulation.  Since $\opt$ is a relaxed benchmark where supply constraints only bind in expectation, this reduction in supply can reduce the value of $\opt$ by a factor of at most $(1-\epsilon/10)$.

Thus, the total welfare obtained by Algorithm~\ref{alg.reduction} is at least
$(1-O(\epsilon))\frac{1}{\alpha}\opt$, as long as $\gamma = O(\epsilon)$.  The latter will be true as long as $C>\Omega( \maxC\,\epsilon^{-2} \log( \epsilon^{-1} + \maxS ) )$, from the definition of $\gamma$.  We conclude that Algorithm~\ref{alg.reduction} is $\alpha(1+O(\epsilon))$-competitive as required.

\section{Robust Menu-Based Schedulers}
\label{sec:oracle}

To complete our solution for \MainProb, we design menu-based, {receptive}, robustly-$\alpha$-competitive algorithms for \RelaxProb, to be used in conjunction with Theorem~\ref{thm:reduction}. We achieve $\alpha = O(\log(\maxD H))$ for the adversarial problem variant (when the entire birth sequence is fixed by an adversary), and absolute-constant $\alpha$ for the stochastic problem variant.%
\footnote{The two problem variants carry over to \RelaxProb in an obvious way.}
For both results, {per-round running time is $\textsc{Poly}(\epsilon^{-1}, \maxS, |\term{support}(P_j)|)$}.
We defer full proofs to the appendix.

%\subsection{Fully Adversarial Environment}
%In this section we design a robust menu-based online algorithm for fractional scheduling, to be used in Theorem~\ref{thm:reduction}.
%Our algorithm will apply in a fully adversarial setting and will be $O(\log(\maxD H))$-competitive.  Our menu-based algorithm is listed as Algorithm~\ref{alg.scheduler}.

%Recall that valuations are assumed to lie in $[1, H]$ for some $H \geq 1$, and to have finite support.  \bledit{Recall that $S_{max}$ is an upper bound on the difference between a job's submission time $B_j$ and the latest start time such that the resulting job's value can be non-zero.}
%(i.e., at most $S_{max}$ possible start times). {\bf[RL: valuations are by job endings, not possible start times, so there could be A-max+D-max-1 values per job ]}
%Also, for ease of notation, we assume that each job is submitted at a distinct time.
%\blcomment{Previously we had defined $H$ as a maximum value density, but for now it seems easier to simply bound the maximum value.  Todo: once we've completed the proofs, revisit and see if we get better bounds by focusing on density rather than absolute value.}

\begin{algorithm}[t]
\alginit
%\caption{\textsc{Online-Fractional-Allocation}}
\caption{Adversarial \RelaxProb}
\label{alg.scheduler}
\setcounter{AlgoLine}{0}
	\SetKwInOut{Require}{require}
	\SetKwInOut{Input}{input}
	\Require{Capacity constraint $C \geq 1$}
%	\Input{Online arrivals of jobs}
	
%	\BlankLine
	\label{alg.frac.line1}For all rounds $t$, initialize $p_t \leftarrow 1/(2\maxD)$, $y_t \leftarrow 0$\;
	\If{\label{alg.frac.line2}some job $j$ arrives at time $t$}{
        Price menu: $\altpay^t(x_{ij}) = \sum_{t'}p_{t'} \cdot x_{ijt'}$\;\label{alg.frac.alloc.choice}
    	Choose some allocation $x_j \in \argmax_{x_j} \; \tilde{V}_j(x) - \altpay^t(x_j)$ \;\label{alg.frac.line3}
        Input adversarially chosen allocation $x_j$ (if applicable)\;
    	\For{each $t' \geq t$}{
            $y_{t'} \leftarrow y_{t'} + x_{j t'}$\;
    	    $p_{t'} \leftarrow (4 H \maxD)^{y_{t'} / C} \cdot (1/(2\maxD))$\;\label{alg.frac.line7}
    	}
	}
\end{algorithm}

\xhdr{Adversarial Variant.}
We present Algorithm~\ref{alg.scheduler}. At each round $t$, it maintains a price per unit of resource at each future round $t' \geq t$.  The price function ${\altpay}^t(x_j)$ is a combination of these per-unit prices: $\altpay^t(x_j) = \sum_{t'}p_{t'} x_{jt'}$.   We then choose a fractional allocation to maximize the expected utility from job $j$.
%Note that since prices and values are linear, we need not distinguish between the full allocation $x_j$ to job $j$ and the breakdown of allocations to each of its tasks, so we will suppress this distinction for the remainder of this section.
%After choosing the allocation,
Subsequently,
each price $p_{t'}$ is then updated as a function of $y_{t'}$, the total (fractional) allocation of resources at time $t'$ (including the job just scheduled). Write $p_{t'} = p(y_{t'})$, where
\begin{equation}
\label{eq.adv.price}
p(y_{t'}) = (4 H \maxD)^{y_{t'} / C} \cdot 1/(2\maxD).
\end{equation}
Note that $p(0) = 1/(2\maxD )$, $p(C) = 2H$, and the prices increase exponentially in usage.  {These values are tuned so that resources are affordable for any job when usage is $0$, but always greater than any customer's willingness to pay when the supply is exhausted.}

%\begin{theorem}\label{thm:welfare}
%If $C = \Omega(\maxC \log(H \maxD))$, the online scheduler described by Algorithm~\ref{alg.scheduler} is robustly $O(\log(\maxD H))$-competitive for the fractional scheduling problem with adversarial arrivals.
%\end{theorem}

\begin{theorem}[adversarial variant]\label{thm:welfare}
If $C > \Omega(\maxC \log(H \maxD))$, Algorithm~\ref{alg.scheduler} is robustly $\alpha$-competitive for \RelaxProb, $\alpha=O(\log(\maxD H))$. Plugging it into
Algorithm~\ref{alg.reduction} with parameter {$\eps$ such that
    $C>\Omega( \maxC\,\epsilon^{-2} \log( \epsilon^{-1} + \maxS ) )$,
we obtain an $O(\eps \maxV)$-truthful,} $O(\alpha)$-competitive algorithm for \MainProb.
\end{theorem}

\begin{proof}[Proof Sketch]
%Our approach to proving Theorem~\ref{thm:welfare} is to
We show that, under our prescribed schedule of price increases, the total value obtained by the algorithm is not much less than the total price of all resources (due to the exponential pricing function), which itself cannot be much less than the total difference in value between the optimum solution {and} the algorithm's solution (since the optimal-in-hindsight allocation is one of the options considered by Algorithm~\ref{alg.scheduler} on Line~\ref{alg.frac.alloc.choice}).
\end{proof}

\vspace{-1mm}\xhdr{Stochastic Variant.}
We recall the definition of {this variant}.
%We next consider the following stochastic model of job arrival for fractional scheduling.
The number of jobs is fixed, but each job $j$'s \SoW is drawn independently from a known distribution $F_j$ of finite support.
%\footnote{Recall that in the fractional scheduling setting, a \SoW consists of a valuation $\tilde{V}_j$, a width $\dem_j$, a collection of tasks $(\tau_i)$, and a corresponding collection of weights $(\lambda_i)$.}
Once all the \SoWs have been drawn, an adversary can choose the submission time for each job, subject to being before the earliest arrival time.

{
\begin{theorem}[stochastic variant]\label{thm:main.stochastic}
Suppose $C > \Omega(\maxC  \log(\maxD))$ and $\sigma(\arr,\dur) = \dur$ (i.e., durations are revealed upon job arrival).
% and job durations are revealed at arrival time.
Then there is a robustly $\alpha$-competitive algorithm for \RelaxProb, where $\alpha$ is an absolute constant. Plugging this into
Algorithm~\ref{alg.reduction} {with parameter $\eps$ such that
    $C>\Omega( \maxC\,\epsilon^{-2} \log( \epsilon^{-1} + \maxS ) )$,
we obtain an $O(\eps \maxV)$-truthful,} $O(\alpha)$-competitive algorithm for \MainProb.
\end{theorem}
}

\begin{proof}[Proof Sketch]
%Our approach to proving Theorem~\ref{thm:main.stochastic} is to
We first solve an LP relaxation that encodes the stochastic version of $\opt$, where supply constraints need only hold in expectation over the distributions $F_j$.  We then need to round this LP solution, online, into a feasible schedule.  For this we use a technique from Chawla et al.\cite{chawla2019} to partition the LP solution (which is a weighted collection of potential allocations) into disjoint sub-solutions, each of which is associated with a small quantity of resources and can be rounded independently.  We then associate each sub-solution with a per-unit price, calculated using the LP solution value, that will be assigned to its corresponding allocations.  Using techniques from Prophet Inequalities~\cite{feldman2014combinatorial,dutting2020prophet,chawla2019}, we show that the posted-price algorithm that allocates in a utility-maximizing way using these prices gives an $O(1)$-approximation to the LP value.
%The full details are omitted due to space constraints.
\end{proof}

\section{Conclusions and Future Work}
%\section{Conclusion}
\label{sec:conc}
This work presented truthful scheduling mechanisms for cloud workloads submitted with uncertainty in jobs'  future arrival time and execution duration. These dimensions of uncertainty model the characteristics of repeated jobs and computation pipelines that are prevalent in production workloads.
%, e.g. of machine learning retraining pipelines.
%The adversarial \SoW arrival model corresponds to cases where the cloud system receives many ad-hoc workloads, as well as recurring workloads with high variance characteristics. The Bayesian arrival model corresponds to cloud systems that receive mostly recurring workloads with good predictability characteristics.
{We show how to approach both adversarial and stochastic variants of this model in a unified framework.  We reduce to a relaxed problem without uncertainty by employing a particular LIFO eviction policy that minimizes the disruption (to both welfare and incentives) when the available resources are over-allocated in hindsight.}
%{\bf \mbc{remove:} TODO summarize approach and cite results.}

Taken literally, our model suggests an interface where customers provide probabilistic information directly to the platform. \ASdelete{We emphasize that} This is an abstraction; a more practical implementation would involve a \emph{prediction engine} implemented internally to the platform that predicts the arrival time and duration distributions of regularly submitted jobs.  We could then view a \SoW as a combination of user-specified input and the predictions, and we would like to ensure that customers are not incentivized to mislead or otherwise confuse the prediction engine. {Making this perspective rigorous runs into a subtle three-way distinction between agents' beliefs, the engine's predictions, and the true distributions; we leave this to future work.}

{Another natural} direction for future work is to extend the analysis to richer workload models. For example, {\em elastic} distributed workloads that may be executed at various concurrency settings, executing faster when utilizing more nodes and slower when running on fewer.  Another extension is to preemptable jobs, whose execution may be paused and later resumed without causing the job to fail.  {Finally, while we focused on obtaining worst-case competitive ratios in this paper, we note that the welfare guarantees in our reduction (Theorem~\ref{thm:reduction}) actually apply per-instance.  It would be interesting to explore whether this translates into improved performance in well-motivated classes of problem instances.
%therefore be interesting to apply our approach to practical scheduling methods on real-world usage traces or simulations of realistic workload specifications, to directly evaluate the benefits of using stochastic predictions to schedule workloads in advance.
}

\bibliographystyle{plain}
\bibliography{amlPricing}

\begin{thebibliography}{10}

\bibitem{agmon2013deconstructing}
Orna Agmon Ben-Yehuda, Muli Ben-Yehuda, Assaf Schuster, and Dan Tsafrir.
\newblock Deconstructing amazon ec2 spot instance pricing.
\newblock {\em ACM Transactions on Economics and Computation (TEAC)},
  1(3):1--20, 2013.

\bibitem{al2013cloud}
May Al-Roomi, Shaikha Al-Ebrahim, Sabika Buqrais, and Imtiaz Ahmad.
\newblock Cloud computing pricing models: a survey.
\newblock {\em International Journal of Grid and Distributed Computing},
  6(5):93--106, 2013.

\bibitem{aspnes1993line}
James Aspnes, Yossi Azar, Amos Fiat, Serge Plotkin, and Orli Waarts.
\newblock On-line load balancing with applications to machine scheduling and
  virtual circuit routing.
\newblock In {\em Proceedings of the twenty-fifth annual ACM symposium on
  Theory of computing}, pages 623--631, 1993.

\bibitem{awerbuch1993throughput}
Baruch Awerbuch, Yossi Azar, and Serge Plotkin.
\newblock Throughput-competitive on-line routing.
\newblock In {\em Proceedings of 1993 IEEE 34th Annual Foundations of Computer
  Science}, pages 32--40. IEEE, 1993.

\bibitem{azar2015truthful}
Yossi Azar, Inna Kalp-Shaltiel, Brendan Lucier, Ishai Menache, Joseph Naor, and
  Jonathan Yaniv.
\newblock Truthful online scheduling with commitments.
\newblock In {\em Proceedings of the Sixteenth ACM Conference on Economics and
  Computation}, pages 715--732, 2015.

\bibitem{azar1997line}
Yossi Azar, Bala Kalyanasundaram, Serge Plotkin, Kirk~R Pruhs, and Orli Waarts.
\newblock On-line load balancing of temporary tasks.
\newblock {\em Journal of Algorithms}, 22(1):93--110, 1997.

\bibitem{babaioff2015dynamic}
Moshe Babaioff, Shaddin Dughmi, Robert Kleinberg, and Aleksandrs Slivkins.
\newblock Dynamic pricing with limited supply, 2015.

\bibitem{babaioff2017era}
Moshe Babaioff, Yishay Mansour, Noam Nisan, Gali Noti, Carlo Curino, Nar
  Ganapathy, Ishai Menache, Omer Reingold, Moshe Tennenholtz, and Erez Timnat.
\newblock Era: A framework for economic resource allocation for the cloud.
\newblock In {\em Proceedings of the 26th International Conference on World
  Wide Web Companion}, pages 635--642, 2017.

\bibitem{bao2018online}
Yixin Bao, Yanghua Peng, Chuan Wu, and Zongpeng Li.
\newblock Online job scheduling in distributed machine learning clusters.
\newblock In {\em IEEE INFOCOM 2018-IEEE Conference on Computer
  Communications}, pages 495--503. IEEE, 2018.

\bibitem{bartal2003incentive}
Yair Bartal, Rica Gonen, and Noam Nisan.
\newblock Incentive compatible multi unit combinatorial auctions.
\newblock In {\em Proceedings of the 9th conference on Theoretical aspects of
  rationality and knowledge}, pages 72--87, 2003.

\bibitem{canetti1998bounding}
Ran Canetti and Sandy Irani.
\newblock Bounding the power of preemption in randomized scheduling.
\newblock {\em SIAM Journal on Computing}, 27(4):993--1015, 1998.

\bibitem{chaudhary2020balancing}
Shubham Chaudhary, Ramachandran Ramjee, Muthian Sivathanu, Nipun Kwatra, and
  Srinidhi Viswanatha.
\newblock Balancing efficiency and fairness in heterogeneous gpu clusters for
  deep learning.
\newblock In {\em Proceedings of the Fifteenth European Conference on Computer
  Systems}, pages 1--16, 2020.

\bibitem{chawla2017truth}
Shuchi Chawla, Nikhil Devanur, Janardhan Kulkarni, and Rad Niazadeh.
\newblock Truth and regret in online scheduling.
\newblock In {\em Proceedings of the 2017 ACM Conference on Economics and
  Computation}, pages 423--440, 2017.

\bibitem{chawla2017stability}
Shuchi Chawla, Nikhil~R Devanur, Alexander~E Holroyd, Anna~R Karlin, James~B
  Martin, and Balasubramanian Sivan.
\newblock Stability of service under time-of-use pricing.
\newblock In {\em Proceedings of the 49th Annual ACM SIGACT Symposium on Theory
  of Computing}, pages 184--197, 2017.

\bibitem{chawla2019}
Shuchi Chawla, J.~Benjamin Miller, and Yifeng Teng.
\newblock Pricing for online resource allocation: Intervals and paths.
\newblock In Timothy~M. Chan, editor, {\em Proceedings of the Thirtieth Annual
  {ACM-SIAM} Symposium on Discrete Algorithms, {SODA} 2019, San Diego,
  California, USA, January 6-9, 2019}, pages 1962--1981. {SIAM}, 2019.

\bibitem{chung2020unearthing}
Andrew Chung, Subru Krishnan, Konstantinos Karanasos, Carlo Curino, and
  Gregory~R Ganger.
\newblock Unearthing inter-job dependencies for better cluster scheduling.
\newblock In {\em 14th $\{$USENIX$\}$ Symposium on Operating Systems Design and
  Implementation ($\{$OSDI$\}$ 20)}, pages 1205--1223, 2020.

\bibitem{curino2014reservation}
Carlo Curino, Djellel~E Difallah, Chris Douglas, Subru Krishnan, Raghu
  Ramakrishnan, and Sriram Rao.
\newblock Reservation-based scheduling: If you're late don't blame us!
\newblock In {\em Proceedings of the ACM Symposium on Cloud Computing}, pages
  1--14, 2014.

\bibitem{devanur2019near}
Nikhil~R Devanur, Kamal Jain, Balasubramanian Sivan, and Christopher~A Wilkens.
\newblock Near optimal online algorithms and fast approximation algorithms for
  resource allocation problems.
\newblock {\em Journal of the ACM (JACM)}, 66(1):1--41, 2019.

\bibitem{dutting2020prophet}
Paul Dutting, Michal Feldman, Thomas Kesselheim, and Brendan Lucier.
\newblock Prophet inequalities made easy: Stochastic optimization by pricing
  nonstochastic inputs.
\newblock {\em SIAM Journal on Computing}, 49(3):540--582, 2020.

\bibitem{emek2020stateful}
Yuval Emek, Ron Lavi, Rad Niazadeh, and Yangguang Shi.
\newblock Stateful posted pricing with vanishing regret via dynamic
  deterministic markov decision processes.
\newblock {\em arXiv preprint arXiv:2005.01869}, 2020.

\bibitem{feldman2014combinatorial}
Michal Feldman, Nick Gravin, and Brendan Lucier.
\newblock Combinatorial auctions via posted prices.
\newblock In {\em Proceedings of the twenty-sixth annual ACM-SIAM symposium on
  Discrete algorithms}, pages 123--135. SIAM, 2014.

\bibitem{ferguson2012jockey}
Andrew~D Ferguson, Peter Bodik, Srikanth Kandula, Eric Boutin, and Rodrigo
  Fonseca.
\newblock Jockey: guaranteed job latency in data parallel clusters.
\newblock In {\em Proceedings of the 7th ACM european conference on Computer
  Systems}, pages 99--112, 2012.

\bibitem{gu2019tiresias}
Juncheng Gu, Mosharaf Chowdhury, Kang~G Shin, Yibo Zhu, Myeongjae Jeon, Junjie
  Qian, Hongqiang Liu, and Chuanxiong Guo.
\newblock Tiresias: A $\{$GPU$\}$ cluster manager for distributed deep
  learning.
\newblock In {\em 16th $\{$USENIX$\}$ Symposium on Networked Systems Design and
  Implementation ($\{$NSDI$\}$ 19)}, pages 485--500, 2019.

\bibitem{jain2015near}
Navendu Jain, Ishai Menache, Joseph Naor, and Jonathan Yaniv.
\newblock Near-optimal scheduling mechanisms for deadline-sensitive jobs in
  large computing clusters.
\newblock {\em ACM Transactions on Parallel Computing (TOPC)}, 2(1):1--29,
  2015.

\bibitem{jain2014truthful}
Navendu Jain, Ishai Menache, Joseph~Seffi Naor, and Jonathan Yaniv.
\newblock A truthful mechanism for value-based scheduling in cloud computing.
\newblock {\em Theory of Computing Systems}, 54(3):388--406, 2014.

\bibitem{jalaparti2016dynamic}
Virajith Jalaparti, Ivan Bliznets, Srikanth Kandula, Brendan Lucier, and Ishai
  Menache.
\newblock Dynamic pricing and traffic engineering for timely inter-datacenter
  transfers.
\newblock In {\em Proceedings of the 2016 ACM SIGCOMM Conference}, pages
  73--86, 2016.

\bibitem{jyothi2016morpheus}
Sangeetha~Abdu Jyothi, Carlo Curino, Ishai Menache, Shravan~Matthur
  Narayanamurthy, Alexey Tumanov, Jonathan Yaniv, Ruslan Mavlyutov, Inigo
  Goiri, Subru Krishnan, Janardhan Kulkarni, et~al.
\newblock Morpheus: Towards automated {SLO}s for enterprise clusters.
\newblock In {\em 12th $\{$USENIX$\}$ Symposium on Operating Systems Design and
  Implementation ($\{$OSDI$\}$ 16)}, pages 117--134, 2016.

\bibitem{koren1992d}
Gilad Koren and Dennis Shasha.
\newblock {\em D-OVER; an optimal on-line scheduling algorithm for overloaded
  real-time systems}.
\newblock PhD thesis, Inria, 1992.

\bibitem{leighton1995fast}
Tom Leighton, Fillia Makedon, Serge Plotkin, Clifford Stein, Eva Tardos, and
  Spyros Tragoudas.
\newblock Fast approximation algorithms for multicommodity flow problems.
\newblock {\em Journal of Computer and System Sciences}, 50(2):228--243, 1995.

\bibitem{lucier2013efficient}
Brendan Lucier, Ishai Menache, Joseph Naor, and Jonathan Yaniv.
\newblock Efficient online scheduling for deadline-sensitive jobs.
\newblock In {\em Proceedings of the twenty-fifth annual ACM symposium on
  Parallelism in algorithms and architectures}, pages 305--314, 2013.

\bibitem{mahajan2020themis}
Kshiteej Mahajan, Arjun Balasubramanian, Arjun Singhvi, Shivaram Venkataraman,
  Aditya Akella, Amar Phanishayee, and Shuchi Chawla.
\newblock Themis: Fair and efficient $\{$GPU$\}$ cluster scheduling.
\newblock In {\em 17th $\{$USENIX$\}$ Symposium on Networked Systems Design and
  Implementation ($\{$NSDI$\}$ 20)}, pages 289--304, 2020.

\bibitem{menache2014demand}
Ishai Menache, Ohad Shamir, and Navendu Jain.
\newblock On-demand, spot, or both: Dynamic resource allocation for executing
  batch jobs in the cloud.
\newblock In {\em 11th International Conference on Autonomic Computing
  ($\{$ICAC$\}$ 14)}, pages 177--187, 2014.

\bibitem{peng2018optimus}
Yanghua Peng, Yixin Bao, Yangrui Chen, Chuan Wu, and Chuanxiong Guo.
\newblock Optimus: an efficient dynamic resource scheduler for deep learning
  clusters.
\newblock In {\em Proceedings of the Thirteenth EuroSys Conference}, pages
  1--14, 2018.

\bibitem{plotkin1995fast}
Serge~A Plotkin, David~B Shmoys, and {\'E}va Tardos.
\newblock Fast approximation algorithms for fractional packing and covering
  problems.
\newblock {\em Mathematics of Operations Research}, 20(2):257--301, 1995.

\bibitem{porter2004mechanism}
Ryan Porter.
\newblock Mechanism design for online real-time scheduling.
\newblock In {\em Proceedings of the 5th ACM conference on Electronic
  commerce}, pages 61--70, 2004.

\bibitem{tumanov2016tetrisched}
Alexey Tumanov, Timothy Zhu, Jun~Woo Park, Michael~A Kozuch, Mor
  Harchol-Balter, and Gregory~R Ganger.
\newblock Tetrisched: global rescheduling with adaptive plan-ahead in dynamic
  heterogeneous clusters.
\newblock In {\em Proceedings of the Eleventh European Conference on Computer
  Systems}, pages 1--16, 2016.

\bibitem{wu2019cloud}
Caesar Wu, Rajkumar Buyya, and Kotagiri Ramamohanarao.
\newblock Cloud pricing models: Taxonomy, survey, and interdisciplinary
  challenges.
\newblock {\em ACM Computing Surveys (CSUR)}, 52(6):1--36, 2019.

\end{thebibliography}

\clearpage
\appendix

\section{Table of Notation}
\label{app:notation}

\medskip

\begin{center}
\begin{tabular}{l|l}
$C$         & system's computational capacity\\
$\dem_j$    &concurrency demand of job $j$\\
$\birth_j$  &birth time of job $j$\\
$\arr_j$    &arrival time of job $j$\\
$\dur_j$    &duration of job $j$\\
$\sigma_j$  &signal about job $j$ duration revealed at arrival time\\
$P_j$       &joint probability distribution over $(\arr_j, \dur_j)$\\
$\finish_j$ &completion/finish time of job $j$\\
$V_j(t)$    & value derived by job $j$ if it completes at time $t$ \\
$\maxC$     & $\max$ concurrency demand of any job.\\
$\maxD$     & $\max$ duration of any job \\
$\maxS$     & $\max$ time difference between birth and completion\\
$\maxV$     & $\max$ value of any job
\end{tabular}
\end{center}

%\section{Omitted Proofs}

%%%%%%%%%%%%%%%%%%%%%%%
% SECTION 3
%%%%%%%%%%%%%%%%%%%%%%%
%\subsection{Omitted Proofs from Section 3}

\section{Proof of Theorem~\ref{thm:incentives}}

%\begin{proof}
{Suppose job $j$ is submitted at time $t$, with \SoW report
    $\SoW_j = (V',\dem',P')$. For now,}
assume
    $\estP^t = \evict^t$
for all $t$.
%Suppose that job $j$ reports concurrency $\dem'$, and that launch plan $L'$ is chosen for job $j$.

If $\dem' < \dem_j$ then the job will receive no value from executing, so we can assume that $\dem' \geq \dem_j$.  Since the job only pays for resources that it uses (and then only if the job successfully completes), %OLD: isn't evicted),
and since prices are set at time $t$, {its expected utility is
    $\estU_t(L_j \mid V_j,\dem',P_j)$, as per \eqref{eq:util-L}.}
%\[ E_{(a,d) \sim P_j}\left[\left(V_j(L'(\arr)+\dur) - \pay^t(L'(\arr),\dem',\dur)\right)\left(1-\evict^t(L'(\arr),\dem',\dur)\right)\right]. \]
%\mbc{Again, here we have the problem that  "$t'=L'(a)$" uses a notation that does not explicitly depends on $j$, see comment above.}
%\[ E_{(a,d) \sim P_j, t' = L'(a)}\left[\left(V_j(t'+d) - c' \cdot \sum_{s = t'}^{t'+d-1} p_s\right)\left(1-{\mathcal E}^{t}(t',c',d)\right)\right]. \]
Note that this utility is weakly decreasing as $\dem'$ increases, since higher $\dem'$ only increases the price
    $\pay^t(L_j(\arr,\sigma(\arr,\dur)),\dem',\dur)$
and (true) failure probability
    $\evict^t(L_j(\arr,\sigma(\arr,\dur)),\dem',\dur)$,
for all arrival times $\arr$ and durations $\dur$.
%, and signals $\sigma$.}
Since $\dem' \geq \dem_j$, it must be utility-maximizing to declare $\dem' = \dem_j$.  So from this point onward assume that $\dem' = \dem_j$.

{The job's expected utility for any given launch plan $L$ is therefore
    $\estU_t(L \mid V_j,\dem_j,P_j)$.
Note that it depends only on the true \SoW, but not on the reported distribution $P'$ nor the value function $V'$.
%otherwise independent of job $j$'s reported distribution and valuation, say $P'_j$ and $v'_j$. We can therefore write $U(L)$ for this expected utility (under the assumption that $\dem' = \dem_j$).  Job $j$'s
The job's utility is therefore maximized when the agent reports truthfully.}

%with any launch plan that maximizes
%    \ASedit{$U(\cdot \mid V_j,\dem',P_j)$.}
%Since such a launch plan is chosen when the agent reports \ASedit{truthfully: $P'=P_j$ and $V'=V_j$},
%$P'_j = P_j$ and $v'_j = v_j$,
%it is utility-maximizing for the agent to do so.

Now suppose that the estimated failure probabilities are potentially incorrect by up to $\mu$ in expectation.  Then the expected calculation of utility for any launch plan with non-negative utility can differ by up to $\mu H$ from the true utility.  Thus the chosen plan can have expected utility up to $2\mu H$ less than the optimal plan, where here the expectation also includes any randomness in the eviction probability estimator.

%%%%%%%%%%%%%%%%%%%%%%%%%
% SECTION 4
%%%%%%%%%%%%%%%%%%%%%%%%%

\section{Proof of Lemma~\ref{lem:eviction}}

We use the following concentration bound, which strengthens the standard Azuma-Hoeffding inequality.  It considers weighted sums of random variables, where the variables and their weights can depend on earlier realizations.  Importantly, the probability bound depends on the expected sum of the random variables, but not the number of random variables.  This is important for Lemma~\ref{lem:eviction}, where we need to establish an error bound that is uniform with respect to time and the number of jobs processed by the algorithm.
Lemma~\ref{lem:concentration} is a variation of a bound that appears as Theorem 4.10 in~\cite{babaioff2015dynamic}.  We omit the proof due to space constraints. %this. \blcomment{This is a variation of a bound that appears as Theorem 4.10 in ``Dynamic Pricing with Limited Supply.'' I couldn't see a way to use that result as-is (specifically because its probability bound has a dependence on $n$), so the following basically reproves it with a slightly different bound that removes the dependence on $n$ at the cost of a worse dependence on $M$.}

%\begin{theorem}[Theorem 4.10 in "dynamic pricing with limited supply"]\label{thm:concentration}
%Let $x_1, \dotsc, x_n$ be $0-1$ random variables.  For each $t$, let $\alpha_t \in [0,1]$ be a multiplier determined by $x_1, \dotsc, x_{t-1}$.  Let $M = \sum_t M_t$, where $M_t = E[\alpha_t x_t | x_1, \dotsc, x_{t-1}]$ for each $t$.  Then for any $\delta > 0$ we have
%\[ \Pr[ \sum_t \alpha_t x_t \geq ]
%\end{theorem}

%\blcomment{Todo: replace the following concentration bound with a martingale version}

\begin{lemma}\label{lem:concentration}
Suppose $x_1, \dotsc, x_n$ are Bernoulli random variables
%with $x_i = c_i$ with probability $p_i$, otherwise $x_i = 0$, where
and that $c_1, \dotsc, c_n$ are real numbers satisfying $0 \leq c_i \leq \maxC$ for each $i$, where $c_i$ and the distribution of $x_i$ can depend on $x_1, \dotsc, x_{i-1}$.  Write $X = \sum_i c_i x_i$ and suppose $E[X] \leq M$.  Then for any $\delta \leq 1$,
\[ \Pr[ X > (1+\delta)M ] < e^{- \Omega( \frac{\delta^2}{1+\delta} \cdot \frac{M}{\maxC}) }. \]
\end{lemma}

With Lemma~\ref{lem:concentration} in hand, we are now ready to prove Lemma~\ref{lem:eviction}.

\begin{proof}[Proof of Lemma~\ref{lem:eviction}]
Recall that $x_{jt}$ is the total expected allocation assigned to tasks of job $j$ in the simulated fractional scheduling problem.
Fix some arbitrary $\tilde{t} \geq t'$.
%For any realization of job arrivals and durations, we know that $\sum_j x_{j\tilde{t}} \leq C'$, since the assignment generated by $\alg$ is feasible for supply $C'$.
Choose an arbitrary assignment of execution plans that satisfy the condition $\sum_j x_{j\tilde{t}} \leq C'$, where the execution plan assigned to each job $j$ can depend on the realization of arrival times and durations of previously-submitted jobs.

%That is, $x_{j\tilde{t}}$ is the expected value of a random variable that is either $0$ or $c(j)$, with the probability potentially depending on the realization of previous jobs and hence on the realized usage of resources at time $\tilde{t}$ $x_{k\tilde{t}}$   Note that this assignment could be adversarial and adaptive; we are not requiring that it be aligned with the behavior of the online algorithm.

Write $z_{j\tilde{t}}$ for the \emph{realized} usage of resources at time $\tilde{t}$ by job $j$.  Then we know that $x_{j\tilde{t}} = E[z_{j\tilde{t}}]$, where the expectation is over the arrival and duration of job $j$, and $z_{j\tilde{t}}$ is either $0$ or $c(j)$.  The distribution of $z_{j\tilde{t}}$ is determined by the launch plan assigned to job $j$, which can depend on the realization of $z_{k\tilde{t}}$ for jobs $k$ that were submitted prior to job $j$.
Therefore Lemma~\ref{lem:concentration} applies to the random variables $\{z_{j\tilde{t}}\}_j$ (considered in the order in which jobs are submitted), and by taking $M = C' < (1-\delta)(C - \maxC)$ we conclude
\[ \Pr\left[ {\textstyle \sum_j} z_{j\tilde{t}} > C - \maxC \right] < e^{- \Omega(\frac{\delta^2}{1+\delta} \cdot \frac{C-\maxC}{\maxC}) }. \]
Write $A(\tilde{t})$ for the event that $\sum_j z_{j\tilde{t}} > C - \maxC$.  We then have that
\begin{equation}
    \label{eq:time.zero.bound}
    \Pr[ A(\tilde{t}) ] < e^{- \Omega(\frac{\delta^2}{1+\delta} \cdot \frac{C-\maxC}{\maxC}) }
    \quad\text{for any fixed $\tilde{t} \geq t'$}.
\end{equation}

We note that the probability bound \eqref{eq:time.zero.bound} is with respect to all randomness in realizations as evaluated at time $0$.
To bound $\evict^t(t',\dem,\dur)$, we instead need to bound the probability of $A(\tilde{t})$ as evaluated at time $t$, conditioned upon the history of all observations (i.e., job arrival and completion events) up to time $t$.  We therefore need to consider the evolution of $\Pr[ A(\tilde{t})]$ from time $0$ to time $t$, then take a union bound over the timesteps $\tilde{t}$ that can impact $\evict^t(t',\dem,\dur)$.  To this end, consider the history of all realizations that occur up to time $t$.  Call this history $\hist$, which is a random variable with a finite support.  We can therefore write $\Pr[A(\tilde{t})] = \sum_\hist \Pr[\hist] \Pr[A(\tilde{t}) | \hist]$.  Now, in preparation for taking a union bound, write $B(\tilde{t})$ for the event that $\Pr[ A(\tilde{t}) | \hist ] > \lambda/(\maxS)$.  We then have that
\[ \Pr[A(\tilde{t})] = {\textstyle \sum_\hist}\; \Pr[\hist] \Pr[A(\tilde{t}) | \hist] > \Pr[B(\tilde{t})] \cdot (\lambda/\maxS), \]
and hence $\Pr[B(\tilde{t})] < \Pr[A(\tilde{t})] / (\lambda/\maxS)$.  In other words,
\begin{equation}
\label{eq.B.bound}
\Pr[B(\tilde{t})] < \frac{\maxS}{\lambda} \cdot e^{- \Omega(\frac{\delta^2}{1+\delta} \cdot \frac{C-\maxC}{\maxC}) }.
\end{equation}
Now consider a job that is submitted at time $t$, requires $\dem$ units of resources each round, and has (realized) duration $\dur$.  Regardless of what schedule this job is assigned, it can be evicted only if the total realized usage exceeds $C - c$ (and hence exceeds $C - \maxC$) in some round between $t$ and the time at which the job was scheduled to complete, which is at most $t + \maxS$.  So by a union bound over the events $\{A(t), A(t+1), \dotsc, A(t+\maxS)\}$ given $\hist$, we have that
\[ \evict^t(t',\dem,\dur) \leq {\textstyle\sum_{k=0}^{\maxS}} \Pr[A(t+k) | \hist]. \]
Thus, in order for $\evict^t(t',\dem,\dur)$ to be larger than $\lambda$, we must have $\Pr[A(t+k) | \hist] > \lambda / (\maxS)$ for at least one choice of $k \in \{0, \dotsc, \maxS\}$, which is to say that at least one of the events in $\{B(t), B(t+1), \dotsc, B(t+\maxS)\}$ occurs.  Taking a union bound over these events and applying \eqref{eq.B.bound} yields the desired bound:
\[ \Pr[\evict^t(t',\dem,\dur) > \lambda] \leq {\textstyle \sum_{k=0}^{\maxS}}\; \Pr[B(t+k)] < \frac{(\maxS)^2}{\lambda} \cdot e^{- \Omega(\frac{\delta^2}{1+\delta} \cdot \frac{C-\maxC}{\maxC}) }. \qedhere \]
\end{proof}

\section{Proof of Theorem~\ref{thm:welfare}}

We first note that the allocation $x_j$ chosen by Algorithm~\ref{alg.scheduler} for job $j$ is always feasible.  To see why, note that if $y_t > C-\maxC$ then $p(y_t) > H$.  But the maximum value attainable by any allocation of any job that consumes $z > 0$ units of computation on round $t$ is $H \cdot z$, which would be less than the price paid for round $t$ only.  We conclude that if $y_t > C-\maxC$ then no further allocation of resources at time $t$ will be made, and hence we will always have $y_t \leq C$.

To bound the competitive ratio of Algorithm~\ref{alg.scheduler}, we will use an argument inspired by dual fitting.  To this end, we will compare the value from the obtained solution $(x_{ij})$ to an appropriate function of the prices.  Note that when job $j$ arrives and is allocated $x_j$, then since the job obtains non-negative utility we have
\[ \tilde{V}_j(x_j) \geq \sum_t x_{jt} p(y_t) \geq \frac{1}{2} \int_{y_t}^{y_t + x_{jt}}p(z) dz \]
where the second inequality follows since $p(z + \maxC) \leq 2 p(z)$ for all $z < C$, as long as $C > \maxC \log(4H\maxD)$.
Write $p_t^*$ and $y_t^*$ for the prices and total usage, respectively, at the conclusion of the algorithm.  Then, summing over all $j$ and integrating the formula in \eqref{eq.adv.price}, \begin{equation}\label{eq:adv.dual}
\sum_j \tilde{V}_j(x_j)
    \geq \frac{1}{2} \sum_t \int_0^{y^*_t} p(z) dz
    = \frac{C}{2\log{4H\maxD}} \sum_t p^*_t - p(0).
\end{equation}
%\blcomment{Todo -- the line above is too terse; we should add a few more sentences.}
Now recall the definition of a robustly competitive algorithm, and let $N$ denote the subset of jobs $j$ that are not adversarially allocated.
Let $\{z_j\}_j$ denote any (possibly randomized) feasible allocation of the jobs in $N$.  For convenience we will write $z_{jt}$ for the expected allocation at round $t$ under $z_j$.  Let $S \subseteq N$ denote the subset of jobs for which $\E[\tilde{V}_j(z_j)] > 3\sum_t z_{jt} (p^*_t - p(0))$, and let $T = N \setminus S$.  For any $j \in S$, since job $j$ could have been allocated any allocation in the support of $z_j$, and since prices at the birth of job $j$ can only be lower than $(p^*_t)$, we conclude from the choice of $x_j$ (on Line~\ref{alg.frac.line3} of Algorithm~\ref{alg.scheduler}) and linearity of expectation that
\begin{align}
\label{eq.adv.S}
\tilde{V}_j(x_j) &\geq \tilde{V}_j(x_j) - \altpay^{\birth_j}(x_j)\nonumber\\
&\geq \E[\tilde{V}_j(z_j)] - \sum_t z_{jt} p^*_t\nonumber\\
&\geq
\E[\tilde{V}_j(z_j)] - \sum_t z_{jt} p(0) - \sum_t z_{jt} (p^*_t - p(0))\nonumber\\
& \geq \E[\tilde{V}_j(z_j)] - \frac{1}{2}\E[\tilde{V}_j(z_j)] - \frac{1}{3}\E[\tilde{V}_j(z_j)]\nonumber\\
& = \frac{1}{6}\E[\tilde{V}_j(z_j)]
\end{align}
where the second-to-last inequality line from the definition of $S$ and the fact that $\E[\tilde{V}_j(z_j)] \geq 2\sum_t z_{jt}p(0)$ since $p(0) = 1/2\maxD$ is half the minimum value density of any job.

Next consider jobs in $T$, and note that we must have $\sum_{j \in T}z_{jt} \leq C$ for each round $t$, by feasibility.  Thus
\begin{align}
\label{eq.adv.T}
\sum_{j \in T} \E[\tilde{V}_j(z_j)] & \leq 3\sum_{j \in B}\sum_t z_{jt} (p^*_t - p(0))\nonumber\\
&\leq 3C \textstyle{\sum_t}\; p^*_t - p(0)\nonumber\\
&\leq 6\log\{4H\maxD\} {\textstyle \sum_j}\; \tilde{V}_j(x_j) \end{align}
where the last inequality is \eqref{eq:adv.dual}.
%\blcomment{Need to be more careful about the $p(0)$ term in \eqref{eq:adv.dual}.  Clean way to handle this?}
Combining \eqref{eq.adv.S} and \eqref{eq.adv.T} yields
\begin{align*}
\sum_j \E[\tilde{V}_j(z_j)] &= \sum_{j \in A}\E[\tilde{V}_j(z_j)] + \sum_{j \in B}\E[\tilde{V}_j(z_j)]\\
&\leq (6\log(4H\maxD)+6) {\textstyle \sum_j} \tilde{V}_j(x_j).
\end{align*}
Thus Algorithm~\ref{alg.scheduler} is robustly $O(\log(H\maxD))$-competitive.
%\end{proof}

%%%%%%%%%%%%%%%%%%%%%%%%%%%%%%%%%%%%%%%%%%%%%
% Stochastic Setting

%\begin{proof}[Proof of Theorem~\ref{thm:main.stochastic}]

\section{Proof of Theorem~\ref{thm:main.stochastic}}

We first recall the statement of Theorem~\ref{thm:main.stochastic}. We suppose $C > \Omega(\maxC  \log(\maxD))$ and $\sigma(\arr,\dur) = \dur$ (i.e., durations are revealed upon job arrival).
% and job durations are revealed at arrival time.
Then we claim that there is a robustly $\alpha$-competitive algorithm for \RelaxProb, where $\alpha$ is an absolute constant. Plugging this into
Algorithm~\ref{alg.reduction} {with parameter $\eps$ such that
    $C>\Omega( \maxC\,\epsilon^{-2} \log( \epsilon^{-1} + \maxS ) )$,
we obtain an $O(\eps \maxV)$-truthful,} $O(\alpha)$-competitive algorithm for \MainProb.
%We will prove this theorem under the additional assumption that when a job arrives, the scheduler also learns that job's duration. We discuss how to remove this assumption in Discussion~\ref{disc:duration}.

We begin the proof by expressing the optimal relaxed fractional assignment as an LP.  We will say that an outcome for job $j$, indexed by $\omega$, is a tuple $(j, \SoW_j, i, x_{ij})$.  The interpretation is that $j$'s realized statement of work (from the support of $F_j$) is $\SoW_j$, and that task $i$ of $\SoW_j$ was provided interval allocation $x_{ij}$.  We will write $\SoW(\omega)$, $x(\omega)$, $\tau(\omega)$, etc., for the $\SoW$, allocation, and task associated with $\omega$, respectively. We will also write $q_{\SoW_j}$ for the probability of statement of work $\SoW_j$ under $F_j$.
%A relaxed (randomized) fractional schedule is then a
%
Our relaxed LP is then as follows, where the variables $(z_{\omega})$ are interpreted as the fractional assignment of each outcome $\omega$.
\begin{align*}
\text{max} & \sum_\omega \tilde{V}(\omega) z_{\omega}\\
\text{s.t.} & \sum_\omega x_{t}(\omega) z_{\omega} \leq C\quad \forall t\\
& \sum_{\substack{\omega:\\\SoW(\omega) = \SoW_j,\\\tau(\omega) = \tau_i}} z_{\omega} \leq q_{\SoW_j} \quad \forall \tau_i \in \SoW_j, \forall \SoW_j \in Supp(F_j), \forall j\\
& z_\omega \in [0,1] \quad \forall \omega
\end{align*}
Here the first constraint imposes the supply restriction, that the total expected resources allocated over all possible outcomes is at most $C$.  The second constraint is that the total probability assigned to outcomes for a given subtask of a given $\SoW$ does not exceed the probability that the $\SoW$ is realized from $F_j$.  For a given solution $z$ to this LP, we will write $Val(z)$ for its total value.

We are now ready to describe our approach to computing prices for the menu-based algorithm promised by Theorem~\ref{thm:main.stochastic}.  For this we will use the notion of a ``fractional unit allocation'' from Chawla et al.\cite{chawla2019}.  We restate it here in our notation.
This involves a slight extension of their definition, since we allow interval allocations to have width up to $\maxC$.

\begin{definition}
An LP solution $(z)$ is a \emph{fractional unit allocation} if there exists a partition of the multiset of resources (where each resource in round $t$ has multiplicity $C$) into bundles $\{B_1, B_2, \dotsc \}$ and a corresponding partition of job outcomes $\omega$ with $z_\omega > 0$ into sets $\{A_1, A_2, \dotsc\}$ such that:
\begin{itemize}
    \item For each $k$ and $\omega \in A_k$, $x(\omega) \subseteq B_k$
    \item For each $k$, $\sum_{\omega \in A_k} w(\omega) z_{\omega} \leq \maxC$
    \item For each $k$ and $t$, if $B_k$ contains any units of resource from round $t$, then $B_k$ contains at least $\maxC$ units from round $t$.
\end{itemize}
\end{definition}
Roughly speaking, a fractional unit allocation can be decomposed into disjoint ``sub-allocations'' that are independent of each other, such that the total fractional weight of each sub-allocation is at most $\maxC$ (the maximum demand of a single job).  The third condition ensures that it is always feasible to schedule any single outcome from each set of the partition.
%We note that the third of the conditions in our definition did not appear in \cite{chawla2019}; it is only necessary because of the rescaling.

We will make use of the following result from~\cite{chawla2019}, which is implicit in the proof of their Theorem 1.2.  We again restate in our notation.\footnote{In~\cite{chawla2019} it was assumed that $\maxC = 1$, but the result extends directly to the case of $\maxC > 1$. Indeed, in the relaxed LP, a task of width greater than $1$ can be treated equivalently as a collection of tasks each with width at most $1$.  And since our bound on total supply is also scaled by $\maxC$, the requirement that $B_k$ contains at least $\maxC$ units or none, in each round, corresponds to the fact in~\cite{chawla2019} that $B_k$ contains at least 1 unit or none.}

\begin{theorem}[\cite{chawla2019}]
\label{thm:unit-alloc}
Suppose $C > \maxC  \log \maxD $.  Then for any instance of the stochastic fractional allocation problem, there is a fractional unit allocation $z$ that is an $O(1)$ approximation to the optimal allocation value.
%an $O(\frac{\log \alpha}{\log \log \alpha})$ approximation to the optimal allocation value, where $\alpha = \maxD ^{2/\log \maxD }$.
\end{theorem}

In~\cite{chawla2019} it is shown how to use the fractional unit allocation from Theorem~\ref{thm:unit-alloc} to design a static, anonymous bundle pricing menu with high welfare guarantee, for the setting of interval jobs with unit width.  That proof makes use of the assumption that all jobs require exactly one unit of resource per unit time.  This does not hold in our setting, since (a) we allow jobs to have width up to $\maxC $, and (more crucially) (b) in our setting, each task $\tau_i$ has its requirements scaled by $\lambda_i$, which can be arbitrarily small.  However, as we now show, it is still possible to define a pricing function that guarantees high total value in expectation, at the cost of inflating the resource requirements by a constant factor.

\begin{lemma}
\label{lem:unit-to-pricing}
For any fractional unit allocation $z$ that is feasible under supply constraint $C$, there exists a robust menu-based algorithm with supply constraint $2C$ whose expected welfare is at least $\frac{1}{2} Val(z)$.
\end{lemma}
\begin{proof}
Let $A_k$ and $B_k$ be the bundles from the fractional unit allocation $z$.  For each $A_k$ define $V(A_k) = \sum_{\omega \in A_k} \tilde{V}(\omega) z_\omega$, and write $W(A_k) = \sum_\omega w(\omega) z_\omega$.  That is, $V(A_k)$ and $W(A_k)$ are the total fractional value and weight, respectively, of allocations in $A_k$. Then for each bundle $B_k$, we will define the price per unit of $B_k$ to be $p_k = \frac{1}{2W(A_k)} V(A_k)$.
%That is, for any $x \subseteq B_k$, the initial price of $x$ is set to $p(x) = w(x) \cdot p_k$.

%Given these per-unit prices,
We will now define our price function $\altpay^t(x_j)$ for interval allocations (which defines our menu-based algorithm).
%Recall that it suffices to set prices for interval allocations, since durations are revealed at arrival time and hence the value of a general allocation is precisely determined by the convex combination of intervals that make up the allocation.
%
For each bundle $B_k$, write $R_k^t$ for the fractional weight of allocations to $B_k$ up to time $t$.  Initially all of these fractional weights are zero; that is, $R_k^0 = 0$ for all $k$. For each $k$, we say that an interval allocation $x_j$ is \emph{feasible for $B_k$ at time $t$}, written $x_j \in \mathcal{F}^t(k)$, if $x_j \subseteq B_k$ and $w(x_j) + R_k^t \leq 2 \maxC$.  We then define $\altpay^t(x_j) = \min_{k : x_j \in \mathcal{F}^t(k)}\{ w(x) \cdot p_k \}$.
%set the price of $x_j$, $\pi(x_j)$, to be $\min_k \{ w(x) \cdot p_k \}$ where the minimum is over all $k$ such that $x_j$ is feasible for $B_k$.
If $x_j$ is not feasible for any $B_k$ then $\pi(x_j) = +\infty$.  Note that these menu prices are weakly increasing in job duration and width, and that these prices only ever increase as more jobs are scheduled.  They are also well-defined even if some jobs are scheduled arbitrarily (subject to feasibility and non-negative utility), as required by robustness.

Now that our algorithm is defined, we first claim that it generates feasible allocations.  Since allocations can only be made to feasible buckets (even adversarially-selected allocations, since non-feasible buckets have infinite price), the schedule will always maintain the property that $R_k^t$ is at most twice $\maxC$ for all $t$.  That is, the total width of all allocations to $B_k$ is at most $2\maxC$, which means that for all $t$ we have $\sum_{(i,j) \text{ allocated to } B_k} x_{ijt} \leq 2\maxC$ which (from the definition of a fractional unit allocation) is at most twice the number of units of time-$t$ resource contained in multiset $B_k$.  Since the sets $B_k$ formed a partition of at most $C$ items per round, we conclude that $\sum_k \sum_{(i,j) \text{ allocated to } B_k} x_{ijt} \leq 2C$, and hence the resulting allocation will be feasible for supply constraint $2C$.

We next show that the expected welfare generated by this menu-based allocation is at least $\frac{1}{2} Val(z)$.  Recall that the total expected welfare is the sum of the total revenue (payments made) and the total utility of all buyers.  For any realization of the jobs' valuations and arrival order, let $Z_k$ denote the event that the total quantity of bundle $B_k$ purchased at the end of the algorithm is at least $\maxC $.  The total payment made by all jobs is then
\[ Rev = \sum_k \Pr[Z_k = 1] p_k \cdot \maxC \geq \frac{1}{2} \sum_k \Pr[Z_k = 1] V(A_k). \]
Now consider the total utility (value minus payments) obtained by all jobs that are not scheduled adversarially.  Each task of each such job will be allocated to a utility-maximizing choice of bundle $k$ for which it is still feasible.\footnote{It is here where we use the assumption that $\sigma(\arr,\dur) = \dur$.  We are allowing the algorithm to allocate each task independently of the other tasks from the same job, which in particular means that tasks with the same arrival time but different runtimes can be scheduled to different start times (and hence runtime is known to the algorithm at submission time).}  Note that if $Z_k$ does not occur, then bucket $k$ will certainly be feasible for any allocation in $A_k$.  For any outcome $\omega$, write $k_\omega$ for the index $k$ such that $\omega \in A_k$.  Then for a job $j$ with realized statement of work $\SoW_j$ and task $\tau_i$, the user will obtain expected utility (denoted $u_j(\SoW_j, \tau_i)$) at least
\begin{align*}
    & u_j(\SoW_j, \tau_i)\\
    & \geq E\left[ \max_{\substack{\omega : \SoW(\omega) = \SoW_j,\\ \tau(\omega) = \tau_i}} \mathbbm{1}[Z_{k_\omega} = 0] \left(\tilde{V}_j(x(\omega)) - w(x(\omega)) p_{k_\omega}\right)^+ \right] \\
    & \geq \frac{1}{q_{\SoW_j}} \sum_{\substack{\omega : \SoW(\omega) = \SoW_j,\\ \tau(\omega) = \tau_i}} \Pr[Z_{k_\omega} = 0] z_\omega \left(\tilde{V}_j(x(\omega)) - w(x(\omega)) p_{k_z}\right)^+
\end{align*}
where the second inequality follows from the feasibility of solution $z$.
%
%\[ E[ \max_z \Pr[Z_{k_z} = 0] \lambda_z w(z) (v_z - p_{k_j})^+ | v_j ] \geq \sum_{1}{q_{v_j}} \sum_z \Pr[Z_{k_z} = 0] \lambda_z (v_z - p_{k_j})^+. \]
Summing over all jobs, \SoWs, and tasks, the total utility obtained by all buyers is at least
\begin{align*}
    Util & \geq \sum_{j, \SoW_j, \tau_i} q_{\SoW_j} u_j(\SoW_j, \tau_i) \\
    & \geq \sum_{j,\SoW_j, \tau_i} \sum_{\substack{\omega : \SoW(\omega) = \SoW_j,\\ \tau(\omega) = \tau_i}} \Pr[Z_{k_\omega} = 0] z_\omega (\tilde{V}(\omega) - w(\omega)p_{k_\omega})^+ \\
    & \geq \sum_k \Pr[Z_k = 0] \sum_{\omega \in A_k} z_\omega (\tilde{V}(\omega) - w(\omega) p_k) \\
    & \geq \sum_k \Pr[Z_k = 0] ( V(A_k) - W(A_k) p_k ) \\
    & = \sum_k \Pr[Z_k = 0] ( V(A_k) - \frac{1}{2} V(A_k) ) \\
    & = \frac{1}{2} \sum_k \Pr[Z_k = 0] V(A_k).
\end{align*}
The result now follows by summing the utility and revenue terms.
\end{proof}

To complete the proof of Theorem~\ref{thm:main.stochastic}, we construct a fractional unit allocation $z$ as in Theorem~\ref{thm:unit-alloc} under constrained supply $C' = C/2$.  We then use this $z$ to construct a robust menu-based pricing method as in Lemma~\ref{lem:unit-to-pricing} for total supply $2C' = C$.  Combining the approximation factors from Theorem~\ref{thm:unit-alloc} and Lemma~\ref{lem:unit-to-pricing}, we conclude that the resulting scheduler is $O(1)$-competitive for the stochastic fractional scheduling problem.

\section{Estimating Failure Probabilities}
\label{sec:approximate}

In Section~\ref{sec:relaxed.reduction} we described Algorithm~\ref{alg.reduction} assuming access to a relaxed failure probability oracle that is $\epsilon_0$-approximate with probability at least $1-\delta_0$, where $\epsilon_0 = \delta_0 = \Theta(\epsilon)$.
%
%In Section~\ref{sec:analysis} we proved Theorem~\ref{thm:reduction} under the
%assumption that failure probabilities are provided by an oracle.  We now remove this assumption.  Algorithm~\ref{alg.reduction} will proceed as before, except that calls to the oracle $\evict^t$ will instead be replaced by a procedure for estimating the appropriate failure probabilities.  Our estimate, $\evictest^t(t',\dem,\dur)$, will be obtained via sampling.  This sampling approach will introduce errors in our utility calculations that we need to account for in our analysis.
%
We now specify the details of this oracle, which we will implement via sampling.  When a job $j$ is born at time $t$, we consider each potential start time $t' \geq t$ and duration $\dur$ for that job.  Recall that the number of such possible pairs $(t',\dur)$ is bounded.  Note that the launch plans of all other jobs are fixed; we only consider variation in the start time of job $j$.  For each possible $(t',\dur)$, we simulate execution of the resulting schedule $T$ times.  In each simulation we realize any residual randomness in the arrival and duration of all jobs that have been born so far, excluding job $j$, and observe whether job $j$ is evicted given that it starts execution at time $t'$ and runs for $\dur$ timesteps.  Importantly, the failure probability for job $j$ is independent of any jobs that arrive after time $t$, due to the LIFO eviction order. So, in each simulation, job $j$ {fails} %is evicted
with probability exactly $\evict^t(t',\dem_j, \dur)$, independently across simulations.  Our estimate, $\evictest^t(t',\dem_j, \dur)$, will be the empirical average over all $T$ simulations.  Taking $T$ sufficiently large, the Hoeffding inequality a union bound over all choices of $t'$ and $\dur$ will imply that our failure probability estimates are sufficiently accurate.  The following lemma makes this precise.

\begin{lemma}
\ASedit{Fix $\epsilon_0>0$ and $\delta_0>0$ and let $T = \log (\maxD \maxS/\delta_0))\frac{2}{\epsilon_0^2}$.
Suppose we take $T$} samples to estimate failure probabilities in the procedure described above, then for each job $j$ the following event occurs with probability at least $1-\delta_0$: $|\evict^t(t',\dem_j, \dur) - \evictest^t(t',\dem_j, \dur)| \leq \epsilon_0$ for all possible start times $t'$ and durations $\dur$ for job $j$.
\end{lemma}
\begin{proof}
\ASdelete{Fix some $\epsilon_0,\delta_0 > 0$ and let $T = \log (\maxD \maxS /\delta_0)/(2\epsilon_0^2)$.}
 We then have that $\evictest^t(t',\dem_j, \dur)$ is the empirical average of $T$ Bernoulli random variables, each with expectation $\evict^t(t',\dem_j, \dur)$.  Then by the Hoeffding inequality,
\[ \Pr[|\evictest^t(t',\dem_j, \dur) - \evict^t(t',\dem_j, \dur)| > \epsilon_0 ] < e^{-2T\epsilon_0^2} = \frac{\delta_0}{\maxD \maxS}. \]
Taking a union bound over all possible choices of $(t',\dur)$ (of which there are at most $\maxS\maxD $) concludes the proof.
%we conclude that with probability $1-\delta_0$ our estimate will be within $\epsilon_0$ of the true failure probability for each $t'$ and $d$.
%In particular, this means that the estimated failure probability of any given launch plan $L$, $E_{(a,d)\sim P_j, t' \sim L(a)}{\mathcal {\bar{E}}}^t(t',c(j), d)$, will likewise be within an additive error $\epsilon_0$ of the true failure probability.
\end{proof}

\end{document}